\documentclass{article}

\usepackage[latin1]{inputenc}
\usepackage{graphicx,epsfig}
\usepackage{amssymb}
\usepackage{amsthm}
\usepackage{color}

\usepackage{fullpage}

\usepackage{latdefs}
\usepackage{flexray}

\newtheorem{proposition}{Proposition}
\newtheorem{theorem}{Theorem}
\newtheorem{lemma}{Lemma}

\newtheorem{definition}{Definition}

\newtheorem{statement}{Statement}

\begin{document}

\title{Formal verification of a time-triggered hardware interface}

 \author{Julien Schmaltz \\ 
 Open University of the Netherlands \\ School of Computer Science \\ Postbus 6401 DL Heerlen, The Netherlands \\  
  email: Julien.Schmaltz@ou.nl}

\date{}

\maketitle

\begin{abstract}
  We present a formal proof of a
  time-triggered hardware interface. The design implements the bit-clock synchronization mechanism specified 
  by the FlexRay standard for
  automotive embedded systems. The design is described at the gate-level.
  It can be translated to Verilog and synthesized on FPGA.
  The proof is based on  a general model of asynchronous communications
  and combines interactive theorem proving 
  in Isabelle/HOL and automatic model-checking using 
  NuSMV together with a model-reduction procedure, IHaVeIt.
  Our general model of asynchronous communications defines
  a clear separation between analog and digital concerns.
  This separation enables the combination of theorem proving
  and model-checking for an efficient methodology.
  The analog phenomena are formalized in the logic of 
  Isabelle/HOL. The gate-level hardware is automatically
  analyzed using IHaVeIt.
  Our proof reveals the correct values of a crucial parameter of the bit-clock synchronization mechanism.
  Our main theorem proves the functional correctness as well as 
  the maximum number of cycles of the transmission. 
 
\end{abstract}

\section{Introduction}

 
  Communications in distributed systems inherently are asynchronous. To cope 
  with clock imperfections different clock synchronization 
  algorithms are required. FlexRay~\cite{PS05} defines a standard for reliable communications
  in  safety-critical automotive applications. In particular, it defines 
  a \emph{bit-clock synchronization} algorithm that guarantees proper bit transmission
  between two independently clocked registers connected \textit{via} a shared bus.
  In this paper, we prove the formal correctness of a hardware interface
  implementing this bit-clock synchronization algorithm.

  \begin{figure}[ht]
    \centering
    \input{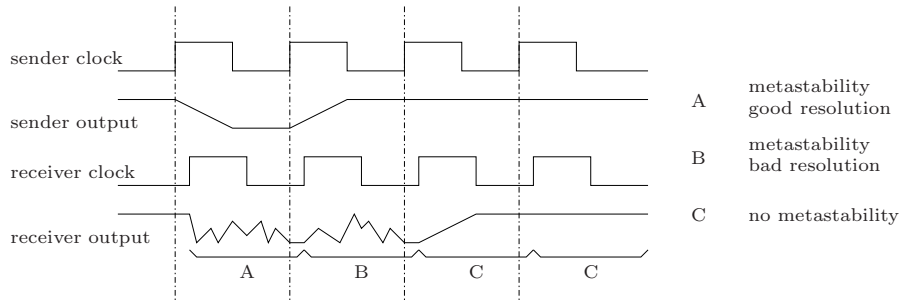}
    \caption{Asynchronous communications and metastability}
    \label{fig:async}
  \end{figure}

  Figure~\ref{fig:async} illustrates one difficulty of interfacing two independently clocked registers\footnote{Our presentation 
    owes a great debt to Moore's introduction~\cite{BIPHIM93}. In particular, Figure~\ref{fig:async} is largely inspired by Figure~2 of Moore's paper.}.    
  Assume a sender and a receiver communicating via a shared bus.
  This picture first shows the sender clock and the signal output on the bus.
  The sender output progressively changes from 1 to 0 and then from 0 to 1.
  In the picture, the receiver
  clock is slightly out-of-phase, i.e., receiver edges appear slightly after sender ones.
  It might be possible for the receiver to sample a signal that is neither a logical 0 nor a logical 1 (See period A and B in Figure~\ref{fig:async}).
  In that case, the receiver reaches a metastable state and ceases to behave as a digital device, i.e.,
  its output is oscillating between high and low voltages. Metastability cannot be avoided~\cite{maenner88}. 
  After the \emph{resolution time}, the receiver output stabilizes 
  to a well-defined value. In the picture, the resolution time is less than a clock cycle. 
  The receiver output stabilizes to 0 (period A) before the end of the cycle. The resolution value is non-deterministic.
  In Figure~\ref{fig:async}, the first metastability resolved to the value sent by the sender but 
  resolution to the negation of this expected value also is possible (see period B).
  For the last two cycles, the sender keeps its output stable and the receiver can always sample 
  a well-defined value. It never reaches a metastable state (periods C).
  In the picture, the clock periods of the receiver and the sender are always equal.
  In practice, clocks suffer from \emph{jitter}. The clock period of one clock is not constant over time, i.e., two successive
  clock cycles will have different lengths. 
  Clocks also suffer from \emph{drift}. 
  The frequencies of two independent clocks are drifting from each other over time.

  The FlexRay interface guarantees proper transmission despite jitter, drift, and metastability.
  The basic idea is that senders keep their output stable long enough to create a \emph{sweet spot} for sampling on the receiver side. 
  We call this stable period a \emph{safe sampling window}. To prevent metastable states, receivers sample bits in the middle of this window.
  If receivers are faster or slower than the sender they will read bits at the beginning or at the end 
  of this window. But if this window is large enough, they will still sample in the region where sender output signals are stable.
  To prove the correctness of our implementation of the FlexRay algorithm, we develop an abstract 
  and formal model of jitter, drift and metastability. This model is general and can be reused in other 
  proof efforts. Our proof shows how to use this abstract model of analog phenomena to reason about digital hardware designs.

  The abstraction of analog phenomena is captured in 
  Proposition~\ref{thm:input-values} (Section~\ref{sec:good}).
  This proposition states precise conditions on the signal produced by the sender. These conditions 
  guarantee successful data transmissions. In Figure~\ref{fig:async}, the last bit
  can be sampled properly because the sender keeps its output signal stable.
  The conditions of Proposition~\ref{thm:input-values} ensure that the sender 
  keeps its output stable long enough to let the receiver sample properly. 
  This proposition mentions analog entities only. 
  Our goal is to analyze \emph{digital} designs. 
  Proposition~\ref{thm:dig-thm} (Section~\ref{sec:comb-world})
  identifies conditions that the sender part of the hardware interface
  must satisfy to ensure proper reception at the receiver part of the hardware interface. These conditions 
  concern \emph{digital} aspects only.
  The formal analysis of hardware designs can abstract away from all analog considerations and stay 
  in the scope of usual automatic verification techniques, e.g., model-checking~\cite{CGP99}. 
  Our main theorem (Theorem~\ref{transcorr}, Section~\ref{sec:proof}) proves that a message of $l$ 
  bytes can be sent and recovered properly using our hardware implementation despite 
  imperfect clocks and asynchronous communications.

  Our model and proof have been developed entirely within the Isabelle/HOL~\cite{IsabelleTutorial}
  theorem prover. 
  Our abstract model of asynchronous communications and the hardware design are represented in the logic of Isabelle. 
  Interactive theorem proving is used to define our abstract model and prove Propositions~\ref{thm:input-values} and~\ref{thm:dig-thm}.
  Properties of the hardware designs are automatically proven using the 
  NuSMV model-checker~\cite{NuSMV02}. NuSMV is used within
  Isabelle \textit{via} a model-reduction interface, named IHaVeIt~\cite{Tverdyshev:TIME08,Tv09}. 
  The synchronization mechanism used in the design is based on resetting a counter when a specific
  sequence of bits is detected. This specific reset value is crucial to the correctness of the 
  algorithm. In the proof of Theorem~\ref{transcorr}, Statements 1 and 2
  identify the exact values ensuring synchronization. This shows that the values 
  proposed in this paper and in the FlexRay standard are correct while
  the value proposed in an early version of our hardware interface~\cite{AutoICCD05} is not.

  In summary, our contribution consists in (1) a clear presentation of a precise model of asynchronous communications;
  (2) the combination of this model with the discrete semantics of hardware design; (3) a hybrid
  verification methodology combining automatic tools with interactive theorem proving; and (4) the proof of 
  the hardware implementation of a time-triggered interface.
  Our proof reveals the specific values of a crucial parameter that ensure
  proper sampling of arbitrary long messages.
  Some of these results have been presented in previous publications~\cite{FMCAD06,FMCAD07}.
  This paper gives a more precise and comprehensible presentation 
  of a unified and extended version of them.

  In the next section we give an overview of our model and its use
  in the verification of the hardware interface.
  The bit-clock synchronization algorithm and its hardware implementation 
  are described in Section~\ref{sec:hardware}.
  The hardware design can be translated to Verilog~\cite{Verilog} and synthesized on FPGA. 
  We present our model of asynchronous communications in 
  Section~\ref{sec:concepts}.  
  This Section presents Proposition~\ref{thm:input-values}.
  We explain the principles of our combination of Isabelle/HOL and IHaVeIt/NuSMV
  in Section~\ref{sec:embedding} and illustrate the derived proof method using a simple example 
  in Section~\ref{sec:proof-example}.

  Section~\ref{sec:proof} 
  proves our correctness theorem by induction on the number of bytes in messages. 
  It shows the values for
  a correct algorithm and gives details about the induction step.
  The difficulty of this proof is that the main theorem states the 
  correctness of the receiver state machine and the synchronization
  hardware. The latter involves reasoning about analog phenomena.
  These two facts are not independent as the hardware controls 
  the state machine and \textit{vice versa}. We need to prove their correctness simultaneously.
  Finally, Section~\ref{sec:related}  discusses related work 
  and Section~\ref{sec:conclu} presents our conclusions.

\section{Overview of our model and our proof}
\label{sec:overview}

  \begin{table}[t]
    \centering
    \begin{tabular}{c|r}
      $\delta$    & Bound on the jitter of all clocks\\
      $\pi$                & Bound on the drift (number of clock cycles) \\
      $c$                   & sender cycle \\
      $\xi$                & receiver cycle \\
      $e_u(x)$            & real-time of the occurrence of edge number $x$ on unit $u$\\
      $\cy(\xi,c)$      & "the mark'' (receiver cycle $\xi$ affected by sender cycle $c$) \\
      $\beta_c^\xi$   & metastability factor (0 or 1 depending on metastability) \\
      $\alpha$           & distance from the mark\\
      $\chi$               & drift factor (-1, 0, or 1) \\
      $t_h$                 & register holding time (real number, \% of receiver clock) \\
      $t_s$                 & register set-up time (real number, \% of receiver clock) \\
      $t_{p_\min}$        & register minimum propagation delay (real number, \% of sender clock)\\
      $t_{p_\max}$       & register maximum propagation delay (real number,  \% of sender clock) \\
      $s(t)$               & value of signal $s$ at real-time $t$ \\
      $\Omega$       & abstract logical value representing signal oscillations\\
      $l[i]$                & element with index $i$ in list $l$ \\
      $\zeta(s,t)$     & conversion to $\{0,1\}$ of signal $s$ at real-time $t$ \\
      $\gamma(l)$   & conversion of list of bits $l$ to a signal taking values in $\{0,1, \Omega\}$\\
      $\clk_u$         & clock of unit $u$ \\
      $\ce_u$          & clock enable signal of output register of unit $u$\\
      $\out_u$        & output signal of output register of unit $u$ \\
      $\inp_u$          &  input signal of input register of unit $u$ \\
      $R_u$             & input or output register of unit $u$ \\
      ${}_aR_u$        & analog register of unit $u$
    \end{tabular}
    \caption{Notations}
    \label{tab:notations}
  \end{table}

  This section gives an overview of our formal model and proof. It introduces principles
  without giving formal definitions. Some notations are mentioned in this section, but 
  only defined later on in the paper. 
  Table~\ref{tab:notations} summarizes the notations used all along this paper.
  
\subsection{Abstract model of asynchronous communications}
  
  Asynchronous communications are facing three issues: clock drift, clock jitter, and 
  metastability. 
  Clock drift denotes the fact that clocks have different frequencies.
  Clock jitter denotes the fact 
  that the frequency of \emph{one particular} clock is not constant over time. This means
  that two consecutive clock cycles may have two different lengths. Finally, registers may sample
  undefined signals and reach metastable states. 
  Our formal model of asynchrony takes these three aspects into account. 
  
  Clock jitter 
  is formalized in Definition~\ref{eq:bcd} (notations $\delta$, Section~\ref{sec:basic-def}).
  The bound on the jitter defines the maximum and minimum
  length of the clock period of all clocks in the system. 
  From this bound on the clock jitter, we derive 
  in Proposition~\ref{eq:bcd-pred} bounds on the drift between two clocks.  
  Our bound
  is expressed as the maximum number of cycles in which the number of clock edges -- 
  called clock \emph{ticks} -- of two clocks may differ by at most one.
  This maximum number of cycles
  is denoted as $\pi$. 
  Given two clocks $u$ and $v$, our bound states that if clock $u$ advances by $\alpha \leq \pi$ clock ticks
  then it is known that clock $v$ will advance by $\alpha$ ticks or $\alpha \pm 1$ ticks.
  Our bound on clock jitter is the same for all clocks of a system. Consequently,
  our bound on the clock drift also is the same for all pairs of clocks.
  Metastability is modeled in the formal definition of \emph{analog} registers (Figure~\ref{fig:ar-def}, Section~\ref{sec:ar}).
  When a register samples a signal that is neither a logical 1 nor a logical 0, its output oscillates before stabilizing.
  Oscillations are represented by an undefined logical value (notation $\Omega$).
  Resolution is represented by a non-deterministic choice between 0 and 1. 
  
  A sender cycle is often referred to as cycle $c$. A receiver cycle is often referred to as cycle $\xi$.
  The time of the rising edge starting cycle $x$ on unit $u$ is noted $e_u(x)$ (Equation~\ref{eq:edge-time}, Section~\ref{sec:basic-def}).
  On a receiver unit, the rising edge $\xi$ that is the closest in time to sender rising edge $c$ is said to be "marked" (or "affected") by $c$, notation $\cy(\xi,c)$
  (See Definition~\ref{def:mark}, Section~\ref{sec:mark}). 
  According to our bound on the clock drift and from a pair of cycles $c$ and $\xi$, we know the mark of any cycle that is less than $\pi$ cycles 
  away from sender cycle $c$. Given a mark $\cy(\xi,c)$ and a distance $\alpha \leq \pi$, the mark for all cycles $c+\alpha$ is known with an error of at most one cycle, 
  i.e., we have $\cy(\xi + \alpha + \chi, c + \alpha)$ with $\chi \in \{ -1, 0, 1 \}$ (Proposition~\ref{thm:reg-affw}, Section~\ref{sec:mark}).
  As edges $c$ and $\xi$ appear approximately at the same time,
  at the time of edge $\xi$, the output signal of the sender output register might not be stable yet -- i.e. still between a logical 0 and a logical 1 -- 
  and  the receiver input register may become metastable. The resolution of this metastable state is a non-deterministic choice that might be 
  the opposite of the value sent by the sender. This resolution to the wrong value might introduce one cycle delay in the receiver input stage.
  Resolution to the good value or no metastable states do not introduce any delay and are treated as a unique case.
  This case distinction is formalized in the \emph{metastability factor}, notation $\beta_c^\xi$ (See Definition~\ref{def:beta}, Section~\ref{sec:meta}).
  The metastability factor returns 1 if a delay is introduced at receiver edge $\xi$ marked by sender cycle $c$ and 0 otherwise.
  Formally, we have $\beta \in \{ 0, 1 \}$.
  
  Finally, the global error is the sum of the error introduced by metastabilities (factor $\beta$) with the error introduced 
  by clock imperfections (factor $\chi$). This sum gives an error in the set $\{ -1, 0, 1, 2\}$. 
  Proposition~\ref{thm:transcorr} (Section~\ref{sec:good}) shows that our implementation of the FlexRay algorithm
  can transmit bits properly even if the number of cycles needed to sample each byte might vary 
  by four cycles.

\subsection{Integration of analog and digital aspects}

  \begin{figure}[h]
    \centering
    \input{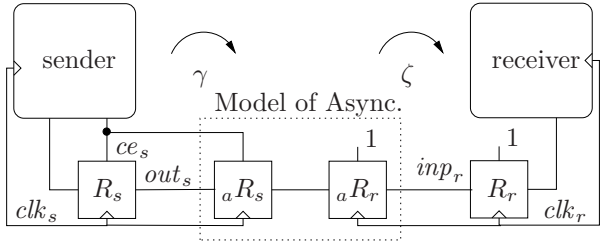}
    \caption{Mixing Analog and Digital Signals}
    \label{fig:mixed-scheme}
  \end{figure}

  We consider the setting pictured in Figure~\ref{fig:mixed-scheme}. 
  The dotted box shows our model of asynchronous communications and two instances 
  of our definition of an \emph{analog} register (Figure~\ref{fig:ar-def}, Section~\ref{sec:ar}).
  Outside this dotted box, the sender and the receiver units as well as their registers connected to the bus 
  correspond to the descriptions made by hardware designers in their favorite 
  hardware description language (e.g., VHDL or Verilog). 
  Registers are composed of a control signal ($\ce$), input and output signals (e.g., $\inp_r$ and $\out_s$), 
  and a clock. 
  In our case, designs are represented 
  in the syntax of Isabelle/HOL. Nevertheless, our description corresponds to synthezisable 
  Register Transfer Level (RTL) designs. 
  The tool IHaVeIt
  can automatically generate
  Verilog code from our Isabelle/HOL syntax~\footnote{More information on the tool can be found at http://www-wjp.cs.uni-saarland.de/ihaveit/.}.
  The idea is to superpose our abstract model above the digital designs. 
  These designs are not modified.
  The purpose of our formal 
  model is to provide an abstraction of the analog phenomena related to asynchronous communications. It identifies 
  constraints on the \emph{digital} units that are sufficient to guarantee proper 
  transmissions in our model of asynchronous communications. This abstraction 
  is captured in Proposition~\ref{thm:input-values} (Section~\ref{sec:good}) and Proposition~\ref{thm:dig-thm} (Section~\ref{sec:comb-world}).
  Proposition~\ref{thm:input-values} identifies the constraints that guarantee proper transmission in our analog model.
  Proposition~\ref{thm:dig-thm} shows which constraints are required on the (digital) sender unit to guarantee proper reception 
  on the (digital) receiver side via our analog model. Proposition~\ref{thm:dig-thm} makes the connection between the digital 
  world of hardware designs and the analog world of asynchronous communications.  
  
\section{Synchronization mechanism and hardware implementation}
 \label{sec:hardware}

 In this section, we introduce the protocol and its hardware implementation.
 We first give an overview of the format of messages and the principles of the protocol.
 We briefly discuss the implementation of sender units. We give more details on the 
 implementation of receiver units.

  \subsection{Protocol overview}
  \label{sec:prot-overview}

  We consider the transmission of bits between an arbitrary 
  number of units connected through a shared bus.
  A basic idea of the time-triggered approach is to give every unit access to the bus 
  during a specific \emph{time slot}. The concatenation of all time slots form a \emph{round} (Figure~\ref{fig:round}).
  Rounds are repeated over and over again. This gives every unit regularly access to the bus.
  During its time slot each unit can send one message.
  Outside its sending slot, a unit listens to the bus waiting for incoming messages. Each unit can send and receive.
  Idle units send a logical one to the bus. At each time, the value on the bus is the conjunction of all the values 
  output by all units. 
  
  \begin{figure}[h]
    \centering
    \input{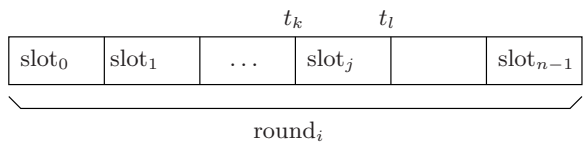}
    \caption{Round and slots}
    \label{fig:round}
  \end{figure}
  
  The division of a round into time slots is a global variable of the entire distributed system.
  To avoid a situation where two units are sending a message at the same time slot, there must 
  be a global understanding on when every slot begins and ends.
  The difficulty is that
  each individual unit is independently clocked and
  each one of them may be at a different time point in a round. It might happen that unit $i$ has its clock
  at the beginning of slot $n$ whereas unit $j$ is still in slot $n-1$, or \textit{vice versa}. 
  One objective 
  of the FlexRay architecture is to maintain the global synchronous abstraction despite the clock imperfections.
  In this paper, we are analyzing a small part of it, namely the \emph{bit-clock synchronization} algorithm.
  This algorithm handles the bit transmission between two independently clocked registers.
  We now describe it.
  The pictures and related explanations are extracted from the FlexRay standard (Chapter 3 Section 3.2.2. of~\cite{PS05}). 

  The principle of the protocol is explained in Figure~\ref{fig:princ1}.
  The first line gives the output of the sender at each clock cycle.
  The second line shows the bit read by the receiver. The last line
  shows the value of a counter maintained by the receiver.
  The counter counts from one to eight.

  The basic idea is to mark the start of the transmission of each byte with a falling edge.
  This falling edge constitutes the byte start sequence $\BSS$ and is created by bits $\BSS[0]$ and $\BSS[1]$. 
  Each bit is sent for eight clock cycles. 
  Figure~\ref{fig:princ1} shows
  the sender output consisting of a byte surrounded by two falling edges.
  The counter is used by the receiver to determine which of the eight copies of a bit should be sampled.
  In the Figure, the receiver samples a bit when its counter equals 5. This sample point
  is called the \emph{strobe} point. The counter is reset to 2 each time the receiver detects a
  falling edge.

  \begin{figure}[h]
    \centering
    \input{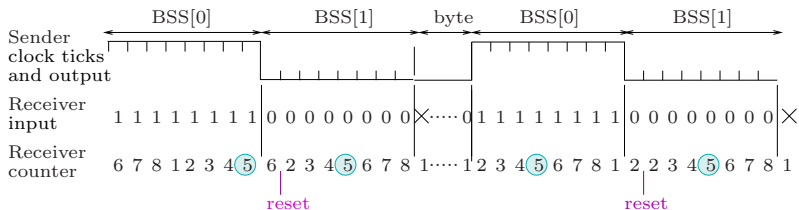}
    \caption{Principle of the protocol}
    \label{fig:princ1}
  \end{figure}

  In the Figure, the counter starts with value six instead of one. This illustrates a situation 
  where the receiver is out of synchronization by three cycles. Because of this delay,
  the receiver samples the last copy of $\BSS[0]$. Then, it detects a falling edge and the counter 
  is reset to two. The receiver samples the fifth copy of the next bit. In the context of perfect clocks,
  the receiver would sample the fifth copy of every bit. 
  After the first falling edge, the receiver misses one bit.
  This is illustrated by the cross replacing a 0 or a 1. This corresponds to a situation 
  where either the receiver clock was too fast and the receiver sampled the last copy of bit $\BSS[1]$
  twice or the receiver was too slow and the receiver will sample the first bit of the byte twice.
  The consequence of these two situations is that the receiver will sample the fourth copy 
  of every bit instead of the fifth one. After detecting the next falling edge and resetting
  the counter, the receiver samples the fifth copy again. So, despite clock imperfections 
  the receiver always starts sampling the fifth copy of every bit. The receiver is kept in synchronization 
  with the sender.

  \begin{figure}[ht]
    \centering
    \input{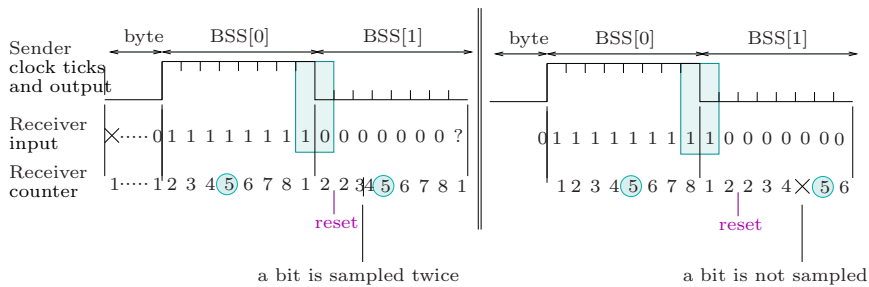}
    \caption{The protocol, drift and metastability}
    \label{fig:princ2}
  \end{figure}

  Figure~\ref{fig:princ2} shows how the protocol works in the presence of drift and metastability.
  Metastability may happen when the receiver samples signals that are transiting between a
  logical 0 and a logical 1 or \textit{vice versa}. 
  As the sender produces eight copies of the same bit, metastability 
  may only take place when sampling the first copy of each bit. 
  The resolution of the metastable state is non-deterministic.
  In the left part in Figure~\ref{fig:princ2}, the receiver reads the correct value of the first bit of $\BSS[1]$.
  This illustrates either that there was no metastability when reading this bit or that metastability resolved to 
  the expected value. 
  The right part in Figure~\ref{fig:princ2} illustrates bad resolution when sampling the first bit of $\BSS[1]$.
  The receiver reads a 1 instead of a 0.
  In Figure~\ref{fig:princ1}, the receiver always reads eight copies of every bit. In practice,
  because of clock imperfections, the receiver might only read seven copies. Formally,
  we can prove that at least seven copies are always read properly (Proposition~\ref{thm:input-values}, Section~\ref{sec:good}).
  The fact that the eighth copy might be misread is pictured by a '?' in Figure~\ref{fig:princ2}. 
  Depending on the effect of metastability these seven copies can be read "early" (left part in Figure~\ref{fig:princ2}) 
  or "late" (right part in Figure~\ref{fig:princ2}).
  
  Our bound on clock jitter and drift is such that missing one cycle in the period starting with the first bit of $\BSS[0]$
  and ending with the last copy of the last bit of the byte is the worst case (Figure~\ref{fig:princ1}). 
  When sampling the following $\BSS$ sequence and byte, a cycle might be missed again. 
  The left part in Figure~\ref{fig:princ2} shows the case where the receiver is \emph{faster}
  than the sender. The counter is updated twice (to 3 and then to 4) when reading only one copy
  of a bit. The consequence is that the receiver will \emph{strobe} earlier and store the fourth copy
  of the bit. The right part in Figure~\ref{fig:princ2} shows the case where the receiver is \emph{slower}
  than the sender. The counter needs two copies of a bit to be updated from 4 to 5. 
  The consequence is that the receiver strobes one cycle later and stores the sixth "good" copy 
  of the bit. 
  The idea is that at least 
  seven copies of every bit will always be stable and ready to be read.
  The objective of the protocol is to strobe one of these seven "good" copies despite drift and metastability.
  The difference between the strobe point and the reset is crucial to the correctness of the protocol.
  Our proof (Section~\ref{sec:lemma2}, statements~\ref{statement1} and~\ref{statement2}) shows that correctness is achieved when this difference is of at least one cycle
  and not greater than three cycles.
  For larger or smaller value of this difference the protocol fails.

  In summary, the main principle of this protocol is to use the $\BSS$ sequence as a "mark" 
  used by receivers to synchronize with the sender.  
  The falling edge of the $\BSS$ sequence "marks" the beginning of a new byte. 
  When a receiver detects that mark it will reset its counter 
  to a specific value. After sampling a byte and because of clock 
  drift the counter of each individual receiver might be slightly different.
  They will detect the next mark with different values of their counter. 
  But, they will all detect the next falling edge and reset their counter to the same value.
  They will all start sampling the next byte with the same value 
  achieving synchronization.


  \subsection{Sender module}

  As idle units put a one on the bus, a sender starts a transmission with a zero.
  This bit is called the transmission start sequence, noted $\TSS$.
  The sender then creates a rising edge by sending another zero and then a one.
  This sequence is called the frame start sequence, noted $\FSS$.
  Before transmitting each byte, the sender starts with the falling edge of 
  the byte start sequence made of $\BSS[0]=1$ and $\BSS[1] = 0$. 
  Finally, the sender ends the transmission with 2 bits creating 
  a rising edge. The last sequence is called the frame end sequence, noted 
  $\FES=01$.
  Let $\tup{a,b}$ be the concatenation of bit vectors $a$ with $b$.
  A message $m$ of $l$ bytes is encapsulated into a frame $f(m)$ with the following format:
  \begin{displaymath}
    \label{eq:frame-format}
    f(m) =  \tup{\TSS,\FSS,\BSS,m[0],\dots,\BSS, m[l-1], \FES}
  \end{displaymath} 
  Each bit of a frame is sent for eight clock cycles.

  \begin{figure}[ht]
    \centering
    \input{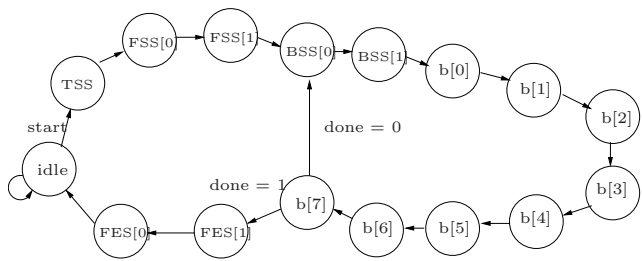}
    \caption{Control Automaton}
    \label{fig:automaton}
  \end{figure}

  The sender embeds bytes into frames by the control automaton in Figure~\ref{fig:automaton}.
  As specified by the protocol, in each state the corresponding bit is generated eight times. 
  The sender is connected with the shared bus through a register named $R_s$ with control enable bit $\ce_s$ (See Figure~\ref{fig:mixed-scheme}).
  This paper focuses on the verification of message reception. We do not detail the 
  sender implementation any further.

  \subsection{Receiver implementation: Bit clock synchronization}
 
  \begin{figure}[ht]
    \centering
    \input{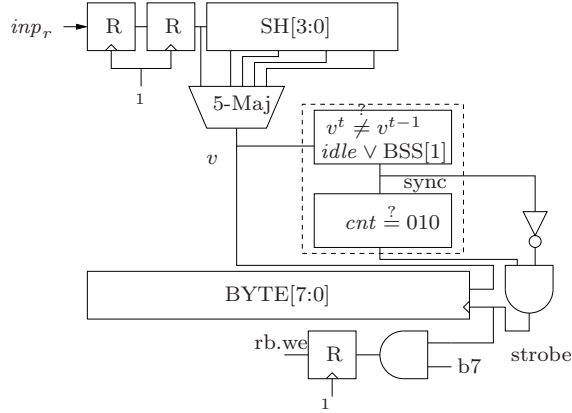}
    \caption{Input Stage}
    \label{fig:recv1}
  \end{figure}
  
  The receiver module implements the same state automaton as the sender. 
  In each state, the receiver is expecting to receive the corresponding bit of the frame eight times.
  Beside the automaton, the relevant part of this receiver consists of the input stage 
  pictured in Figure~\ref{fig:recv1}.
  The first two registers form a ``synchronizer'' used to remedy to metastability.
  Designers used a 2-stage synchronizer, which means that they assume that the resolution time 
  of the metastability is less than one clock cycle.
  The results presented here would be equally applicable to synchronizers of any length.
  A five majority vote is performed.  
  Signal $\sync$ is used to detect the synchronization sequence $\BSS$. It is high 
  if and only if the current voted bit does not equal its previous value and the state automaton is 
  either in state {\it idle} or in state $\BSS[1]$.  When $\sync$ is high counter $\cnt$ is reset to $000$ in the next cycle.
  Let $s^t$ be the value of signal $s$ at hardware
  cycle $t$. Let $z$ denote the state of the receiver automaton.
  Signal $\sync$ is defined by the following Equation:
  \begin{equation}
    \label{eq:sync}
    \sync^t \equiv v^t \neq v^{t - 1} \wedge (z^t = \BSS[1] \vee z^t = \idle)
  \end{equation}
 
  Counter $\cnt$ is defined as follows:
  \begin{equation}
    \label{eq:cnt}
    (\cnt^{t + 1}  = \cnt^t + 1 \wedge \lnot \sync^t) \vee \cnt^{t + 1} = 000  
  \end{equation}

  The state automaton is clocked by signal {\it strobe}, which is high each time 
  the counter reaches value $010$ and the automaton is not 
  synchronizing, i.e., when signal $\sync$ is low. 
  The formal definition of signal $\strobe$ is as follows:
  \begin{equation}
    \label{eq:strobe}
    \strobe^t \equiv \cnt^t = 010 \wedge \lnot \sync^t
  \end{equation}

  Each time {\it strobe} is high, the voted bit is stored in shift register BYTE.
  When the last bit has been stored (i.e., automaton is in state b[7]) and 
  signal {\it strobe} is high,
  signal $rb.we$ turns high and BYTE is written to the main receiver buffer. 
 
  Our implementation differs slightly from the FlexRay guidelines. 
  The standard suggests to reset the counter to $010$ and to strobe when it reaches $101$.
  We reset to $000$ and strobe at $010$. 
  The parameter crucial to the algorithm is the difference between the strobe and the reset values.
  We chose to reset to $000$ because it is slightly simpler to implement than a reset to $010$. 
  In our configuration, the difference between the strobe and the reset values is of
  two cycles. In the FlexRay standard, the difference is of three cycles. 
  In a previous implementation of this algorithm~\cite{AutoICCD05}, 
  the counter is reset to $000$ and $\strobe$ is high when $\cnt$ is $100$.
  In this configuration, the difference between the strobe and reset points is of four cycles.
  One cycle more than in the FlexRay standard. We prove that the synchronization 
  algorithm works only if this difference is of at least one cycle and not greater than
  three cycles.

  \section{Asynchronous communications and the main statement}
  \label{sec:concepts}
  
  This section presents our formal model of asynchronous communications. We first 
  define signals and clocks. After that, we define our bounds on clock jitter and drift.
  After defining the metastability factor and analog registers, the section concludes
  with the correctness of the bit transmission (Proposition~\ref{thm:input-values}).

  \subsection{Signals and clocks}
  \label{sec:sig-clk}

  Time is represented by the nonnegative reals ($\Rp$).
  We assume a finite number of electronic control units (abbr. $\ecu$)
  that are connected through a shared bus.
  The set of all the units is noted $\U$.

  A signal $s$ is represented by a function $s(t)$ from {\it real} time $t$ to $\{0, 1, \Omega\}$:
  1 and 0 mean ``high'' and ``low'' voltages; $\Omega$ means any \emph{other} voltage. 
  Value $\Omega$ abstracts in one logical value all voltages that cannot be identified as a logical 1 or 0.
  Formally, signals have the following functionality:
  \begin{displaymath}
    s: \Rp \to \{0, 1, \Omega \}
  \end{displaymath}

  Because of their cyclic behavior, clocks are not represented by signals but by their period.
  The clock period of unit $u$ is noted $\tau_u$.
  This represents the ideal case. In practice, clock periods suffer from jitter and are not constant over time.
  Jitter is introduced hereafter in Section~\ref{sec:basic-def}.
  Periods are different from zero.
  The {\it time} of the $c^{\mathit{th}}$ rising edge of clock $\clk_u$ of unit $u$ is 
  given by function $e$. Formally, $e$ is a function which converts discrete 
  time to real time relative to the clock of unit $u$. 
  \begin{equation}
   \label{eq:edge-time}
    e: \mathbb{N} \times \U \to \Rp
  \end{equation}
  Function $e$ is defined as the product of $c$ with the clock period:
  $e(c,u) = c\cdot \tau_u$.
  To simplify our notation, we shall write $e_u(c)$ instead of $e(c,u)$.

  A clock cycle is defined by the time interval between two rising clock edges.
  Clock cycle $c$ at unit $u$ is represented by 
  interval $] e_u(c) : e_u(c + 1)]$. The interval is left open to represent the fact 
  that the cycle starts when the clock edge has reached value 1.
  
\subsection{Clock jitter and clock drift}
  \label{sec:basic-def}

  Function $e$ gives the ideal time of edges.
  In practice, 
  clocks suffer from \emph{jitter} and the length of a clock period is not constant over time.
  We assume that all clock periods of any clock deviate at most by a fraction $\delta$ of a reference clock period.
  This reference clock is named $\clk_\tref$. Its period is $\tau_\tref$.
  \begin{definition}
        \label{eq:bcd}
        {\bf Bounded Clock Jitter.}
  \begin{displaymath}
    \Gamma_u \equiv 1 - \delta \leq \frac{\tau_u}{\tau_\tref} \leq 1 + \delta
  \end{displaymath}   
  \end{definition}

  We are not interested in the deviation at each cycle, but in the number of cycles in which the number of 
  ticks of two independent clocks may differ by at most one.
  Let $\pi$ be that number. 
  In this interval, the maximum drift between two clocks is obtained between the slowest and the fastest clocks allowed by 
  our bound on the clock jitter (Definition~\ref{eq:bcd}).
  We derive a bound on the clock drift from the ratio between the minimum and the maximum clock periods.
  From the bound on the clock jitter (Definition~\ref{eq:bcd}) and choosing 
  $\pi = \frac{1-\delta}{2\cdot \delta}$, we prove the following proposition:
  \begin{proposition} {\bf Bounded Clock Drift.}
    \begin{displaymath}
      \label{eq:bcd-pred}
      \Gamma_i \wedge \Gamma_j \rightarrow \frac{\pi}{\pi + 1} \leq \frac{\Min(\tau_i, \tau_j)}{\Max(\tau_i, \tau_j)} 
    \end{displaymath}
  \end{proposition}
  This property is preserved for any number less than $\pi$.

\subsection{Metastability factor}
  \label{sec:meta}

  Our model of metastability links three parts: (1) the undefined voltage $\Omega$, (2) the
  non-deterministic resolution of metastable states to a well-defined value, and (3) resolution to the  
  negation of the expected value. Points (1) and (2) are related 
  in the formal definition of \emph{analog} registers (Figure~\ref{fig:ar-def}).
  The last point is captured by the \emph{metastability factor}, 
  notation $\beta_c^\xi$ (Definition~\ref{def:beta}). We now 
  specify the behavior of analog registers and formally define 
  the metastability factor. Then, we continue with the formal definition 
  of analog registers.

  \begin{figure}[ht]
    \centering
    \input{simple-timdia.pstex_t}
    \caption{Behavior of the register w.r.t clock edge $c$}
    \label{fig:simple-timdia}
  \end{figure}

  Registers consist of one input signal $\In$, one clock signal $\clk$, one control signal $\ce$,
  and one output signal $\Out$. 
  Figure~\ref{fig:simple-timdia} illustrates the behavior of a register.
  A new value ($x$) is input to the register at cycle $c$ (interval $]e_u(c):e_u(c+1)]$). 
  During minimum propagation delay $t_{p_\min}$ the output signal equals previous value $y$.
  Because the control signal is high, the output oscillates (i.e., is $\Omega$) before 
  stabilizing at new value $x$.
  If the control signal is low, the output does not oscillate and keeps its old value $y$.

  If the input or the control signals do not have a constant value during the {\it setup time} 
  (noted $t_s$) before edge $c$
  or during the {\it holding time} (noted $t_h$) after edge $c$, the register may become metastable.
  This means that its output may still be $\Omega$ after $t_{p_\max}$.
  After resolution of this 
  metastability, the receiver input register will output either the value 
  sent by the sender or its negation. The former case 
  is equivalent to the case when there is no metastable state.
  Therefore, we always assume resolution to the negation of the expected input. 
  This case distinction is represented by the \textit{metastability factor} ($\beta$).
  Metastability can only happen if an edge -- say $\xi$ -- (minus the setup time) appears while the 
  sender output is undefined, i.e., before $e_s(c) + t_{p_\max}$. 
  In this case, the metastability factor returns 1. It returns 0 otherwise.
  Formally, the \textit{metastability factor} is a function, 
  which takes as arguments cycles $\xi$ and $c$, and two clocks.
  \begin{definition}\label{def:beta} \textbf{Metastability Factor.}\\[0.1cm]
    \verb+     +$\beta(\xi, c, clk_s, clk_r) \triangleq$ 
    $\aif e_r(\xi) - t_s \leq e_s(c) + t_{p_\max} \athen 1\ \aelse 0$
  \end{definition}
  To alleviate the notation, we shall write $\beta_c^\xi$ instead of $\beta(\xi,c,\clk_s, \clk_r)$.
  The notation $\beta_c^\xi$ denotes whether sampling the bit sent at sender cycle $c$ is affected
  by a potential metastability at receiver cycle $\xi$.

  \subsection{Formal definition of analog registers}
  \label{sec:ar}
  
   A signal $s$ is stable during time interval $[t_1 : t_2]$ if it holds the value at time $t_1$ 
   until time $t_2$. A signal $s$ has a defined value during time interval $[t_1 : t_2]$ if it never equals 
   $\Omega$ during that interval. Formally, this is expressed as follows\footnote{Note: $\stadep$ means
     {\bf sta}ble, {\bf de}fined, {\bf p}redicate}:
   \begin{displaymath}
     \stadep (t_1,t_2,s) \triangleq  \exists b \in \{0,1\}, \forall t \in [t_1:t_2], s(t) = b
   \end{displaymath}

\begin{figure}
  ${}_aR_u(c,\clk_u,\ce_u, \In_u, \Out_u^0) \triangleq$ \\
  \aif $c = 0$ \athen $\lambda t. \Out_u^0$ \aelse \\
  \verb+ +\aif $ \left\{
  \begin{array}{cl}
    & \stadep(e_u(c) - t_s, e_u(c) + t_h, \ce_u) \\ \wedge & \stadep(e_u(c) - t_s, e_u(c)+ t_h, \In_u)
  \end{array} \right. $
  \athen ;; stable inputs -- no metastability

  \verb+  +\aif $\ce_u(e_u(c)) = 1$ \athen ;; update with new value \\
  \verb+  +$\lambda t. \left\{ 
    \begin{array}{lcl} 
      {}_aR_u(c-1, \dots )(e_u(c)) & : & t \in e_u(c) + ]0: t_{p_\min}] \\ 
               \Omega                            & : & t \in e_u(c) + ]t_{p_\min} : t_{p_\max}] \\
               \In_u(e_u(c))                    & : & t \in e_u(c) + ]t_{p_\max} : \tau_u] \\
               \Omega                            & : & t\  \notin\ e_u(c) + ] 0 : \tau_u]\  \mbox{;; to make function total}
    \end{array} \right. $\\
  \verb+  +\aelse ;; keep old value\\
  \verb+  +$\lambda t. \left\{ 
    \begin{array}{lcl} 
      {}_aR_u(c-1, \dots )(e_u(c)) & : & \forall t \in \ e_u(c) + ] 0 : \tau_u] \\ 
                \Omega                            & : & t\  \notin\ e_u(c) + ] 0 : \tau_u] \  \mbox{;; to make function total}
    \end{array} \right.$ \\
  \verb+  +\aendif \\
  \verb+ +\aelse ;; metastability -- non-deterministic resolution to 0 or 1\\
  \verb+ +$\lambda t.\left\{ 
                           \begin{array}{lcl}
                             {}_aR_u(c-1, \dots)(e_u(c)) & : & t \in e_u(c) + ]0: t_{p_\min}]  \\                              
                             \Omega                                  & : & t \in e_u(c) + ]t_{p_\min} : \tau_u - t_s [ \\
                             x \in \{0,1\}                           & : & t = e_u (c + 1) + [-t_s : 0 ]\\
                             \Omega                                  & : & t\  \notin\ e_u(c) + ] 0 : \tau_u]\  \mbox{;; to make function total}
                           \end{array} \right.                $ \\
  \verb+ +\aendif\\
  \aendif
  \caption{Definition of Analog Registers}
  \label{fig:ar-def}
\end{figure}

  The formal definition of the analog behavior is given by function ${}_aR_u$ (Figure~\ref{fig:ar-def}).
  We are interested in the output value of a register
  for all real times during cycle $c$.
  Function ${}_aR_u$ takes as arguments a cycle $c$, a clock signal $\clk_u$, a clock enable control signal $\ce_u$, 
  an input signal $\In_u$, and the initial output value $\Out_u^0$. 
  It generates a signal.

  If no setup or holding time violation occurs, the register behaves normally.
  If the control signal is low, the register keeps its old value (at the previous cycle $c-1$); if the control signal is high
  the output keeps its previous value during $t_{p_\min}$, then oscillates (i.e., is $\Omega$) to finally reach its final 
  value at time $e_u(c) + t_{p_\max}$. 
  If input signal $\In_u$ or control signal $\ce_u$ is not stable and defined during interval $e_u(c) + [- t_s : t_h]$, 
  the register becomes metastable. The output equals the previous
  computation until $t_{p_\min}$ (included) and $\Omega$ afterwards. 
  At the end of the cycle, metastability has been resolved and the output equals an arbitrary but defined value.
  To make the function total, $\Omega$ is output for all times outside the cycle.  
  To alleviate our notation, we shall write ${}_aR^c_u$ instead of 
  ${}_aR_u(c,\clk_u,\ce_u, \In_u, \Out^0_u)$.

  Formally, all timing parameters ($t_h, t_s, t_{p_\min}, t_{p_\max}$) are real numbers expressed
  as percentages of the local clock period.
  In the remainder of this paper, if not precise otherwise, 
  propagation delays are relative to the sender clock and
  setup and holding times are relative to the receiver clock period.
  We assume that the sum of these parameters is less than 1.

  \subsection{The "mark"}
  \label{sec:mark}
  
\begin{figure}[ht]
  \centering
  \input{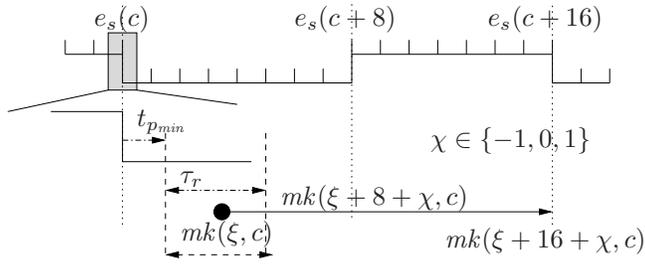}
  \caption{Relating Receivers and Senders}
  \label{fig:cy-fig}
\end{figure}

  The relation between a sender and a receiver is pictured in Figure~\ref{fig:cy-fig}. 
  A sender starts sending three different bits at edges $c$, $c + 8$, 
  and $c + 16$. Each bit is sent for eight clock cycles.
  If we take a closer look around edge $c$, the sender output is not 
  modified before $e_s(c) + t_{p_\min}$, when it moves from $y$ to $\Omega$ 
  (see Figure~\ref{fig:simple-timdia} for more details).
  If a receiver samples before that time, it will get the old value. 
  It is not yet {\it affected} by the new transfer. 
  In contrast, sampling {\it strictly after} that time will affect the receiver, either it will become metastable, 
  or it will detect a new value. 
  At most, it will take a receiver a full cycle to sample after this time.
  Let $\xi$ be the first receiver edge after $e_u(c) + t_{p_\min}$. 
  As this edge is the first one to be affected by the behavior of the sender, we denote it as 
  "marked with edge $c$'', noted $\cy(\xi,c)$. 
  If there is no ambiguity, we may drop the first argument.
  We name this edge "the affected cycle''. It is formally defined as follows:
  \begin{definition}\label{def:mark} {\bf Affected Cycle.}  $\cy(\xi,c) \equiv e_r(\xi) + t_h \in e_u(c) + ]t_{p_\min}:\tau_r ]$
  \end{definition}

  Suppose that 
  edge $\xi$ is affected by some cycle $c$ at which a sender puts a new bit on the bus.
  If the sender sends another bit 
  within a number of cycles ($\alpha$) less than our bound $\pi$, the corresponding 
  affected cycle may be seen by the receiver with a potential 
  error of one cycle, i.e., at $e_r(\xi + \alpha \pm 1)$.
  This means that subsequent marks are known with the same error.
  We name $\chi \in \{-1,0,1\}$ the \emph{drift factor}.
  Figure~\ref{fig:cy-fig} shows these marks for $\alpha = 8$ and 
  $\alpha = 16$. 
  Formally, we have the following Proposition:
  \begin{proposition} {\bf More Affected Cycles}
    \label{thm:reg-affw}
    \begin{displaymath}
      \begin{array}{l}
        \Gamma_r \wedge \Gamma_s \wedge 0 < \alpha \leq \pi \wedge \cy(\xi,c) 
        \rightarrow   \bigvee_{\chi \in \{-1, 0, 1\}} \cy(\xi + \alpha + \chi, c + \alpha)
      \end{array}
    \end{displaymath}
  \end{proposition}
  \begin{proof}
    We do a case analysis depending on the position of $\xi + \alpha$ 
  regarding the receiver cycle expected to be affected by sending at sender cycle $c+\alpha$.
  The expected affected cycle should be in the interval $e_u(c + \alpha) + ]t_{p_\min}:\tau_r ]$.
  If $e_r(\xi+\alpha)$ is (1) before that interval, we prove that it contains
  $e_r(\xi+\alpha + 1)$;
  (2) within that interval, this proves the obvious case where $\chi = 0$; 
  (3) after that interval, we prove that it contains
  $e_r(\xi+\alpha - 1)$.
  \end{proof}

  Proposition~\ref{thm:reg-affw} is important because it gives us which marks can be deduced from the knowledge of  
  a single one. In most of the proofs done in the analysis of the hardware, we always assume only 
  one mark. Then, we use Proposition~\ref{thm:reg-affw} to obtain subsequent marks and 
  perform a case analysis on the three possible times of these marks.

  \subsection{Correctness of asynchronous communications}
  \label{sec:good}

  To ensure that the receiver will not always sample $\Omega$'s, the sender keeps its output 
  constant for several cycles (say $k$ cycles). 
  If $k$ is large enough there exists a "sweet spot'' in which the receiver 
  can sample safely.
  Formally, 
  the safe sampling window of length $k$ w.r.t. cycle $c$ (noted $\SSW^c_k$) is 
  defined as follows: 
  \begin{definition} {\bf Safe Sampling Window.} 
    \begin{displaymath}
      \SSW^c_k \equiv ]e_\clk(c) + t_{p_\max} : e_\clk(c+k+1) + t_{p_\min}]
    \end{displaymath}  
  \end{definition}

  We prove that under our drift hypothesis, $\SSW_k^c$ entails up to $k-1$ receiver cycles (or $k$ edges), even in case of metastability.
  This shows the number of "good'' samples that guarantee reception by the receiver without metastability.
  \begin{proposition} \textbf{SSW's are large enough.}
    \label{thm:affcy}
    \begin{displaymath}
      \begin{array}{l}
        \Gamma_r \wedge \Gamma_s \wedge \cy(\xi, c) \wedge n + 1 \leq k \leq \pi 
        \rightarrow   \forall l \leq n, e_r(\xi+\beta_c^\xi + l) + [ - t_s : t_h] \in \SSW_k^c
      \end{array}    
    \end{displaymath}
  \end{proposition}
  The proposition reads as follows. The first two terms of the hypothesis state that jitter on 
  the receiver and sender clocks is bounded. The third one states 
  that the bit sent at sender cycle $c$ corresponds to receiver 
  cycle $\xi$. The last term assumes that $k$ is not greater than our bound $\pi$ on the clock drift.
  The conclusion shows that the time of $n$ receiver edges together with their 
  set-up and holding times -- hence $n + 1 = k$ cycles -- is within the safe sampling window. 
  Note that there are $k$ cycles even in the presence of metastability ($\beta = 1$).

  Proposition~\ref{thm:transcorr} below proves that sampling in a safe sampling window is correct.
  This first line assumes that jitter is bounded, sender cycle $c$ is related to receiver cycle $\xi$,
  and that $k$ is not greater than our bound $\pi$ on the clock drift.
  The second line expresses the fact that the sender creates a safe sampling window of length $k$.
  The third and fourth assumptions state that
  control and input bits must be stable and defined during interval $e_s(c) + ] -t_s : t_h]$ to avoid metastabilities on the sender side.
  The last two assumptions state that there is a connection between the sender and the receiver and 
  that the time of receiver edge $\xi$ is in the safe sampling window ($\SSW_k^c$).
  The conclusion shows that the output of the receiver register equals the bit sent by the sender.
  
 \begin{proposition} {\bf Correct Transfer.}
   \label{thm:transcorr}
   \begin{displaymath}
     \begin{array}{cl} & \Gamma_r \wedge \Gamma_s \wedge \cy(\xi,c)   \wedge  c > 0 \wedge n + 1 \leq k \leq \pi   \mbox{ (*bounded drift, affected cycle*)}\\   
             \wedge 
                    &  \ce_s(e_s (c)) = 1 \wedge \forall l \in [1:k], \ce_s(e_s(c +l) = 0 \mbox{ (* $\SSW^{c}_k$ *)} \\ 
            \wedge & \forall l \in [0:k + 1], 
             \stadep(e_s(c + l) - t_s, e_s(c + l) + t_h, \In_s) 
             \mbox{(*input *)}\\ 
             \wedge & \forall l \in [0:k + 1], 
             \stadep(e_s(c + l) - t_s, e_s(c + l) + t_h, \ce_s))) \mbox{(*control*)}\\ 
             \wedge 
                    & \forall c, In_r = {}_aR_s(c, \clk_s, \ce_s, \In_s, \Out_s^0) \wedge \forall t, \ce_r(t) = 1 
                    \mbox{ (*analog connection*)}\\ 
             \wedge & e_r(\xi) + [- t_s: t_h] \in \SSW_k^c \mbox{ (* good cycle *)}\\
        \rightarrow 
                    & {}_aR^{\xi}_r (e_r(\xi + 1)) = In_s(e_s(c))
    \end{array}  
  \end{displaymath}
  \end{proposition}
  \begin{proof}
    First, Proposition~\ref{thm:affcy} gives us the position of receiver edges in the safe sampling window.
    Then, we case split on the position of interval $e_r(\xi) + ] - t_s : t_h]$.
    We set two reference points: $e_s(c+1)$ and $e_s(c+1+k)$.
    We prove the conclusion for 5 cases depending on the 
    position of  interval $e_r(\xi) + ] - t_s : t_h]$ regarding these points. 
  \end{proof}

  Finally, Proposition~\ref{thm:input-values} hereafter combines Proposition~\ref{thm:affcy} and Proposition~\ref{thm:transcorr} to prove 
  that for all edges in the safe sampling window the receiver 
  register samples properly. The five hypotheses equal the first five hypotheses of 
  the previous proposition. The conclusion shows that the output value of the 
  receiver register equals the value sent by the sender at cycle $c$ for $x$ cycles.
  Cycle $\xi + \beta_c^\xi$ denotes the first "good" sample after either a bit inversion
  introduced by wrong resolution of a metastable state ($\beta_c^\xi = 1$) or
  a proper reading at receiver cycle $\xi$ ($\beta_c^\xi = 1$). 
  Function ${}_aR^{\xi + \beta_c^\xi + x}(t)$ represents the output signal of the receiver register
  during each "good" cycle $x$. We consider the value at the end of each cycle,
  i.e., at time $e_r(\xi + \beta_c^\xi + x + 1)$.
  
  We illustrate Proposition~\ref{thm:input-values} for the FlexRay protocol 
  and its seven "good" values sketched in Section~\ref{sec:prot-overview} in Figures~\ref{fig:princ1} and~\ref{fig:princ2}.
  The FlexRay protocol specifies that senders must send a bit for eight clock cycles, i.e.,
  they keep their output stable for seven extra cycles. So, we have $k = 7$ and $n = 6$.
  The receiver always reads seven good copies, for $x = 0$ to $6$.
  Depending on the value of $\beta_c^\xi$, these seven good copies are read "early" ($\beta_c^\xi = 0$) or "late" ($\beta_c^\xi = 1$).

 \begin{proposition} {\bf Known Inputs.}
   \label{thm:input-values}
   \begin{displaymath}
     \begin{array}{cl} & \Gamma_r \wedge \Gamma_s \wedge \cy(\xi,c) \wedge c > 0 \wedge n + 1 \leq k \leq \pi \mbox{ (*bounded drift, affected cycle*)}\\   
             \wedge 
                    &  \ce_s(e_s (c)) = 1 \wedge \forall l \in [1:k], \ce_s(e_s(c +l) = 0 \mbox{ (* $\SSW^{c}_k$ *)} \\ 
            \wedge & \forall l \in [0:k + 1], 
             \stadep(e_s(c + l) - t_s, e_s(c + l) + t_h, \In_s) 
             \mbox{(*input *)}\\ 
             \wedge & \forall l \in [0:k + 1], 
             \stadep(e_s(c + l) - t_s, e_s(c + l) + t_h, \ce_s))) \mbox{(*control*)}\\ 
             \wedge 
                    & \forall c, In_r = {}_aR_s(c, \clk_s, \ce_s, \In_s, \Out_s^0) \wedge \forall t, \ce_r(t) = 1 
                    \mbox{ (*analog connection*)}\\ 
        \rightarrow 
                    & \forall x \in [0:n]: {}_aR^{\xi + \beta_c^\xi + x}_r (e_r(\xi + \beta_c^\xi + x + 1)) = In_s(e_s(c))
    \end{array}  
  \end{displaymath}
  \end{proposition}
  \begin{proof}
    Proposition~\ref{thm:affcy} gives us $n+1$ cycles in the safe sampling window. For each one of them 
    we conclude using Proposition~\ref{thm:transcorr}. 
  \end{proof}

  This last Proposition will be rephrased in the next section to mix analog and digital worlds.
  It is key because it gives us which inputs are correctly sampled by the receiver when knowing
  the cycle at which the sender has put a bit on the bus. The "digitalized" version of this formula will 
  convert the conclusion to mention bits and not signals.

\section{Continuous model and discrete semantics}
\label{sec:embedding}

  Our model of asynchronous communications mentions analog entities only.
  The semantics is based on functions and a dense representation of time.
  Ultimately, we want to use this model to verify hardware designs described 
  in another semantics based on a discrete notion of time and 
  transition functions. Before describing our approach, we define type conversion functions and 
  rephrase Proposition~\ref{thm:input-values} to match bits and not signals.

 \subsection{Principle and soundness}

  We recall Figure~\ref{fig:mixed-scheme} that illustrates our integration of our analog results in 
  the analysis of digital designs.                                 
  Our model of asynchrony is shown inside the dashed box. The remainder of the Figure corresponds to digital 
  designs that are actually used to synthesize hardware. These designs are not modified. Our model is simply inserted as 
  a filter of the receiver inputs. Functions $\gamma$ and $\zeta$ converts bits to signals and signals to bits. We precise their definition
  in the next subsection.

  Digital designs are represented by their transition function, one application of which represents the computation of one clock cycle.
  The sender and the receiver parts are analyzed separately. The analysis of the sender does not need any analog arguments.
  It mainly consists of the proof that sender output $\out_s$ follows a specific frame format. The analysis of the receiver
  is done assuming correctness of the sender and that the connection of receiver input $\inp_r$ is done through our model of 
  asynchrony.
  We write that an element $s_u$ of unit $u$ has bit-value $x$ at cycle $c$ 
  - i.e., after $c$ applications of the transition function - as $s_u^c = x$. 
  Formally, we assume that 
  the value of input bit $\inp_r$ at hardware cycle $c$ equals 
  the output value of register ${}_aR_r$ at the time of edge $c+1$:
  \begin{equation}
    \label{eq:mix-ad}
    \forall c, \inp_r^c = \zeta({}_aR_r^{c}, e_r(c + 1))
  \end{equation}
  The left hand side represents the value that should be in register $R_r$ at $c + 1$.
  As the analog register is not part of the transition function of the receiver, 
  one application of the latter compensates this difference. 
  The right hand side is always a defined value. 

  \subsection{Mixing bits and signals}

  Function $\gamma$ is not given any particular definition.
  We only assume that it produces a signal such that during the metastability window around cycle $i+1$
  it outputs the value with index $i$ in the bit list. This property is defined by predicate $\bv2sp$:
  \begin{displaymath}
   \begin{array}{l}
     \bv2sp(\gamma, l_u) \equiv  \forall t,i, t \in e_u(i + 1) + ] - t_s + t_h] \rightarrow \gamma(l_u) = l_u[i]    
   \end{array}
 \end{displaymath}  

 Function $\zeta$ takes as input a signal and a time. 
 If the value of the signal at that time is a bit value, this value is returned. Otherwise, 
 a non-deterministic choice is made and {\it some} bit value is returned.
 \begin{displaymath}
   \zeta(s,t) \triangleq \aif s(t) \in \{0,1\}\ \athen s(t)\ \aelse x \in \{0,1\}
 \end{displaymath}
 
 \subsection{Combining two worlds}
 \label{sec:comb-world}

  Let lists $\ce_s$ and $\In_s$ be the bit lists containing values given to the analog sender register
  ${}_aR_s$.
  If they both satisfy predicate $\bv2sp$, list element $\ce_s[c - 1]$ or $\In_s[c-1]$ corresponds
  to the bit value given to the sender analog register at time $e_s(c)$. 

  Proposition~\ref{thm:input-values} is embedded into a digital context in the following statement.
  We assume that (a) clock drift is bounded; (b) function $\gamma$ correctly translates 
  bit lists $\ce_s$ and $\In_s$; (c) the digital control bits are high once and then low $k$ times to create
  a safe sampling window.
  Analog hypotheses are concerned with the connection of the sender with the receiver and the clock drift. 
  Obviously, they cannot be ``digitalized''.
  These assumptions will be used in almost all theorems, lemmas and propositions proved in the remainder of this
  paper. We denote them by~$\H$, which is formally defined as follows:
  \begin{displaymath}
    \H \equiv \left\{ \begin{array}{cl}
             & \Gamma_r \wedge \Gamma_s \wedge n + 1 \leq k \leq  \pi \wedge \cy(\xi,c) \mbox{ (*bounded drift, mk(c)*)}\\
      \wedge & \forall c, \In_r = {}_aR_s(c, \clk_s, \gamma(\ce_s), \gamma(\In_s), \Out_s^0) \wedge \forall t, \ce_r(t) = 1 
      \mbox{(*link*)}\\
      \wedge & \bv2sp(\gamma, ce_s, \clk_s) \wedge \bv2sp(\gamma, \In_s, \clk_s) \mbox{ (*modeling hypotheses*)}\\
      \wedge &  \ce_s[c + \alpha - 1] = 1 \wedge c > 0 
      \wedge  \forall l \in [1:k], \ce_s[c + l - 1] = 0  \mbox{(*sender*)}\\
      \end{array} \right.
  \end{displaymath}

  Under these assumptions, we prove in Proposition~\ref{thm:dig-thm} below that the "digitalized" output of the {\bf analog} receiver register equals 
  the {\bf digital} input of the sender at cycle $c$. 
  Comparing to Proposition~\ref{thm:input-values}, this proposition differs 
  in its hypotheses and its conclusion. As shown above, hypothesis $\H$ mentions primarily digital entities.
  The left hand-side of the conclusion of Proposition~\ref{thm:dig-thm} is the application of conversion function $\zeta$ to 
  the conclusion of Proposition~\ref{thm:input-values}. The right hand-side is a bit instead of a signal.

  \begin{proposition} {\bf Back to the Digital World.}
    \label{thm:dig-thm}
    \begin{displaymath}
      \begin{array}{cl}
             \H
      \rightarrow \forall x \in [0:n]:
                  
                  \zeta ({}_aR^{\xi + \beta_c^\xi + x}_r, e_r(\xi + \beta_c^\xi + x + 1)) = \In_s[c - 1]
      \end{array}
    \end{displaymath}
  \end{proposition}
  \begin{proof}
    By definition of predicate $\bv2sp$, $\gamma(\In_s)$ and $\gamma(\ce_s)$ are $\stadep$
    for the required cycles. Proposition~\ref{thm:input-values} concludes. 
  \end{proof}

\section{A proof example: correct voted bits}
\label{sec:proof-example}
  To prove that bytes are sampled correctly, we need to prove that each sampled bit is 
  correct, i.e., that the value of the voted bit is correct. This is a very simple
  lemma which illustrates the combination of the continuous time model and the 
  discrete time model, as well as the combination of Isabelle/HOL with IHaVeIt and NuSMV. 
  In this section, we first describe how hardware designs are described and verified
  in the IHaVeIt environment. Then, we show how to incorporate 
  \emph{digital} properties in our model of asynchronous communications. 



\subsection{The IHaVeIt environment}

IHaVeIt stands for {\bf I}sabelle {\bf H}ardware {\bf V}erification {\bf I}nfrastructure 
and has been developed by Tverdsyshev~\cite{Tv09}. It is written in Standard ML and 
implemented as an oracle proof tactic in Isabelle/HOL.
IHaVeIt provides a connection to external verification tools: the NuSMV and SMV model-checkers, and different SAT solvers. The environment 
also provides a tool to generate Verilog descriptions that are then synthesized on FPGA.
The main contribution of this tool is an efficient model-reduction algorithm.
This algorithm is based on a combination of transformation and domain reduction techniques.
These techniques provide data reduction and elimination of functions and memories. Details 
on these algorithms fall outside the scope of this paper (See~\cite{Tv09} for more details).

\subsubsection{Hardware description in Isabelle/HOL}

IHaVeIt considers a subset of the Isabelle/HOL syntax that is suitable to describe hardware, i.e., 
descriptions can be translated to Verilog.  The IHaVeIt subset considers the following basic types:
Boolean, bit vectors, naturals, integers, lists, functions, finite enumeration, and records. Infinite types 
are shrunk using predicate sets. IHaVeIt provides a library of predicate sets of the aforementioned basic types,
e.g., $\mathit{bv\_n(n)}$ and $\mathit{arr\_of(n,t)}$ define the sets of bit vectors of length $n$ 
and arrays of $n$ elements of type $t$. 
Combinatorial circuits are represented using Isabelle/HOL expressions, non-recursive functions and 
uninterpreted functions. Functions are constructed using Isabelle/HOL operators.
Uninterpreted functions are translated to Verilog modules without bodies.
They are typically used to represent memories. Sequential circuits
are represented by standard Mealy machines. State components are stored in registers which
constitute a specific type used in the translation tool.

\begin{figure}[ht]
  \centering
  \input{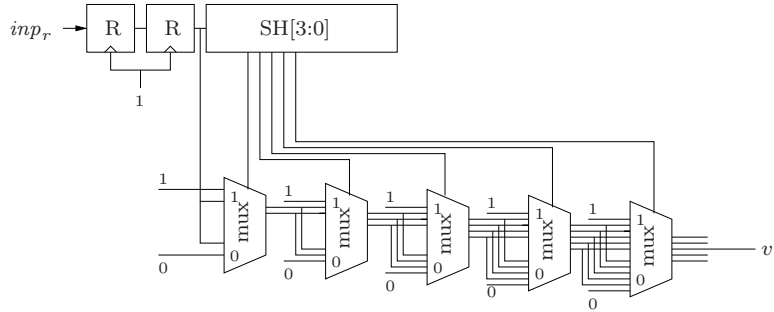}
  \caption{RTL description of majority vote}
  \label{fig:maj}
\end{figure}

We shortly illustrate hardware description in Isabelle/HOL on the very simple example of the computation 
of the voted bit noted $v$ in Figure~\ref{fig:recv1}. Figure~\ref{fig:maj} shows the schematics.
The state component of this circuit is defined by a record containing the two input 
registers and a 4-bit shift register. Using Isabelle syntax, we have\footnote{\texttt{RH} denotes the second input register and \texttt{SH4} the 4-bit shift register in Figure~\ref{fig:maj}.
An 'r' indicates that the element is part of the receiver interface.}:
\begin{verbatim}
record t_rBUSCON = 
      rR    :: t_bitreg
      rRH   :: t_bitreg
      rSH4  :: t_shiftreg4
\end{verbatim}
Majority is computed using a cascade of multiplexers (noted \emph{mux}).
A multiplexer is defined 
as a function taking as arguments two bit-vectors and a select bit.
It returns the selected bit-vector.
\begin{verbatim}
 constdefs mux_impl :: "bv => bv => bit => bv"
    "mux_impl xs ys s == (if (bit2bool s) then xs else ys)"
\end{verbatim}
In computing the majority, multiplexers are used to introduce a 0 or 1.
\begin{verbatim}
 constdefs major_help_impl :: "bv => bit =>bv"
    "major_help_impl b sel == mux_impl (b@[1]) ([0]@b) sel"  
\end{verbatim}
All the multiplexers are connected together to compute the 5-bit majority voting.
\begin{verbatim}
 constdefs major5_impl :: "bv => bit"
    "major5_impl b == 
       let
         v0::bv = [(nth b 0)];
         v1::bv = major_help_impl v0 (nth b 1);
         v2::bv = major_help_impl v1 (nth b 2);
         v3::bv = major_help_impl v2 (nth b 3);
         v4::bv = major_help_impl v3 (nth b 4)
       in 
         (nth v4 2)" 
\end{verbatim}
Finally, the voted bit is computed as the majority of
the four values stored in the shift-register and the value 
stored in the second input register. 
\begin{verbatim}
 constdefs s_v :: "t_rBUSCON => bit"
    "s_v rbuscon == 
        let
          rh_bit::bit = rRH_read_impl rbuscon;
          shift_dout::bv = rSH4_read_impl rbuscon
        in 
          (major5_impl ([rh_bit]@shift_dout))"
\end{verbatim}

\subsubsection{Theorems in IHaVeIt}

IHaVeIt supports the proof of combinatorial properties and temporal properties expressed either in linear or branching time temporal logic (LTL or CTL).
A combinatorial property is a Boolean expression where all free variables are quantified over their subtype, e.g.,
$\forall a \in \mathit{arr\_of(4,bv\_n(8)):}P(a)$. Typically, combinatorial properties are given to a SAT solver or any other kind of decision 
procedures. 
The syntax and semantics
of the LTL and CTL formulas supported by IHaVeIt can be found in Tverdyshev thesis~\cite{Tv09}. 
These formulas correspond to the usual properties described in standard textbooks (e.g., \cite{CGP99,BK08}).
The formalization is inspired by the case-study "Verified Model Checking''
in the Isabelle/HOL tutorial~\cite{IsabelleTutorial}. 

On the circuit computing the majority voting we prove that if the input bit is equal to bit value $b$ for seven clock cycles,
the voted bit equals $b$ for seven cycles with a delay of four cycles. Let $X$ denote the LTL next operator and 
$X^n (P)$ be defined as $X^n (X^{n-1}(P))$. This property is expressed by the following formula:
  \begin{displaymath}
    \begin{array}{ccl}
                           & & \inp_r = b  \wedge X (\inp_r = b ) \wedge X^2(\inp_r = b )\\
                           & \wedge & X^3(\inp_r = b ) \wedge X^4(\inp_r = b ) \wedge X^5(\inp_r = b ) \wedge X^6(\inp_r = b ) \\
\Longrightarrow & \\
                          & &  X^4(v = b ) \wedge  X^5(v = b ) \wedge  X^6(v = b ) \wedge  X^7(v = b ) \wedge  X^8(v = b )  \\
                          & \wedge & X^9(v = b ) \wedge  X^{10}(v = b ) 
    \end{array}
  \end{displaymath}
  Let $\mathit{CorrVotedBit(\inp_r,v)}$ denote this property.
  Let $K = (S,I,T)$ denote the Kripke structure representing the circuit where
  $S$ is the set of states represented by record \texttt{t\_rBUSCON}, $I$ is 
  a predicate defining the set of initial states, and $T$ is the transition function.
  In our case, we make no assumption on the initial states of the registers and $I = \mathit{True}$.
  The transition function is not detailed. It basically applies function \texttt{s\_v} to compute 
  the new value of the voted bit. It shifts all registers by one place to the right and inserts 
  value $\inp_r$ in the first register.
  We prove that property $\mathit{CorrVotedBit}$ holds \emph{always}. 
  \begin{proposition}
    \label{thm:voted-bit-LTL}
    $K \models_{ltl} \Box \mathit{CorrVotedBit(\inp_r,v)}$    
  \end{proposition}
  \begin{proof}
    This property is automatically proven by IHaVeIt that applies model reduction and calls NuSMV. 
    To reduce the state space, we prove it for each possible values of $b$ ($0$ or $1$). 
  \end{proof}

  To use the above temporal property we translate it back to usual logic using the 
  semantic description in Isabelle/HOL. Globally ($\Box$) means that the property holds for all positions
  of all traces. Let $t$ denote an arbitrary position in a trace. Then, 
  $\Box \mathit{CorrVotedBit(\inp_r,v)}$ translates to $\forall t. \mathit{CorrVotedBit(\inp_r^t,v^t)}$.
  Let $\Box X^n (s)$ translate to $\forall t. s^{t + n}$. Then, Proposition~\ref{thm:voted-bit-LTL}
  translates to Proposition~\ref{thm:voted-bit} below:
 \begin{proposition}
    \label{thm:voted-bit}
    $\forall x \in [0:6]: \inp_r^{t + x} = b \rightarrow  v^{t + x + 4} = b$
  \end{proposition}
  \begin{proof}
    This proposition is the translation of Proposition~\ref{thm:voted-bit-LTL} according to the LTL semantics. 
  \end{proof}

  In the next section, we show how to combine this property with our model of asynchronous communications and connect
  with the sender unit. 

\subsection{A simple example: correct voted bits}
\label{ex:voted-bits}

\begin{figure}[ht]
  \centering
  \input{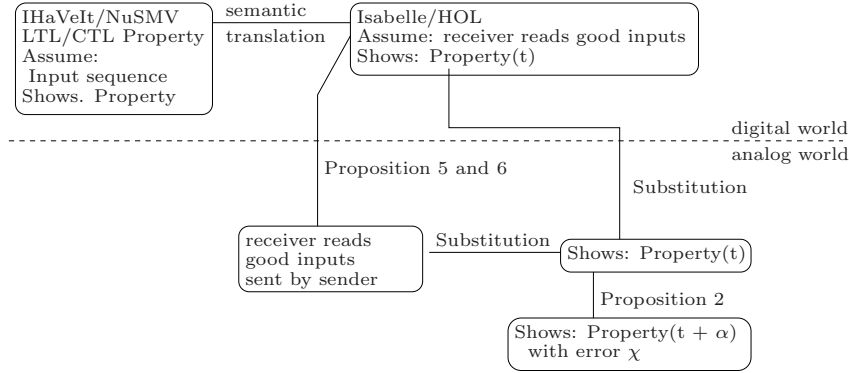}
  \caption{Proof method}
  \label{fig:proof-method}
\end{figure}


  Our general proof method is illustrated in Figure~\ref{fig:proof-method}.
  The first step is to prove a temporal property on the design. In the example of the voted bits, 
  this was done in Proposition~\ref{thm:voted-bit-LTL}. The structure of this property is very 
  typical in our proof. We always prove a property on part of the design assuming a specific input sequence.
  The second step is to translate this temporal property in the logic of Isabelle/HOL using the 
  semantics. This was the purpose of Proposition~\ref{thm:voted-bit}. The property is now 
  dependent of an arbitrary position $t$ in traces. Until now, the \emph{digital} part 
  of the \emph{receiver} unit was analyzed. We now show how to include the sender 
  and our model of asynchronous communications.

  We recall Figure~\ref{fig:mixed-scheme} showing the connection of sender, our model of asynchronous communications,
  and receiver. Formally, this connection is expressed by Equation~\ref{eq:mix-ad}.
  In Proposition~\ref{thm:voted-bit}, we assume that the receiver reads an arbitrary bit value $b$.
  In reality, this bit value is put on the bus by the sender unit at cycle $c$. The receiver is then expected 
  to read the sender output register, i.e., $\out_s^c$ in Figure~\ref{fig:mixed-scheme}.
  The third step of our proof method is to show that the receiver can indeed read this value 
  in our model of asynchronous communications. In our example of the voted bits, 
  we have to show that the receiver can read seven copies of each bit. This is the case 
  if the sender creates a safe sampling window. This is exactly the statement of Propositions~\ref{thm:input-values}
  and~\ref{thm:dig-thm}. Using these propositions, we can discharge the assumptions of our property 
  in the analog world.
  In the FlexRay algorithm, $k$ and $n$ have values 7 and 6 in Proposition~\ref{thm:dig-thm}. 
  The hypotheses obtained by instantiating $k$ and $n$ in $\H$ with these values is noted 
  $\H[k \triangleleft 7,n \triangleleft 6]$. Finally, we obtain the following proposition:
  \begin{proposition}
      \label{eq:known-input}
      $\H[k \triangleleft 7,n \triangleleft 6] \wedge \mbox{Equation ~\ref{eq:mix-ad}}
      \rightarrow \forall x \in [0:6]: \inp_r^{\xi + \beta_c^\xi + x} = \out_s^c$
  \end{proposition}
  \begin{proof}
    Follows directly from Proposition~\ref{thm:dig-thm} and Equation~\ref{eq:mix-ad} 
    by substitution. 
  \end{proof}
  By combining this Proposition with Proposition~\ref{thm:voted-bit}, 
  we obtain the correctness of the voted bit assuming asynchronous 
  transmission, metastability, and clock drift.
  \begin{proposition}
    \label{prop:voted-bit1}
    $\H[k \triangleleft 7,n \triangleleft 6] \wedge \mbox{Equation ~\ref{eq:mix-ad}}
    \rightarrow \forall x \in [4:10]: v^{\xi + \beta_c^\xi + x} = \out_s^c$
  \end{proposition}
  \begin{proof}
    Follows from Propositions~\ref{eq:known-input} and~\ref{thm:voted-bit}, by 
    substituting $t$ by $\xi + \beta_c^\xi + 4$. 
  \end{proof}
  
  Typically, we obtain at this point a formula with a time reference at receiver cycle $\xi$ (or sometimes simply $t$).
  This reference corresponds to the mark, i.e., the association of a receiver cycle with a sender cycle (e.g., $\cy(\xi,c)$).
  We are often interested in generalizing this formula to subsequent marks at a distance $\alpha \leq \pi$ from the known one. 
  This is exactly the purpose of Proposition~\ref{thm:reg-affw}. 
  In the proof of the correct voted bit, we obtain the following proposition:
  \begin{proposition}
    \label{prop:voted-bit2}
    \begin{displaymath}
     \H[k \triangleleft 7,n \triangleleft6] \wedge \alpha \leq \pi \wedge \mbox{Equation ~\ref{eq:mix-ad}}
    \rightarrow 
    \bigvee_{\chi \in \{-1,0,1\} } \forall x \in [4:10]: v^{\xi + \beta_c^{\xi +\theta} + x + \theta} = \out_s^{c + \alpha} 
    \end{displaymath}
    where $\theta = \alpha + \chi$
  \end{proposition}
  \begin{proof}
    Proposition~\ref{thm:reg-affw} gives us three possible marks, one for each value of $\chi$. 
    For each one of them, we instantiate Proposition~\ref{prop:voted-bit1}. 
  \end{proof}

  \section{Formal proof}
  \label{sec:proof}

  In this section we give details on the formal correctness proof of the time-triggered 
  interface. We focus on the receiver correctness and therefore assume correctness 
  on the sender part. We first precise this assumption. Then, we give an overview of the global proof structure.
  After that, we give our correctness statement
  and show the proof of two important lemmas: Lemma~\ref{lemma1} and Lemma~\ref{lemma2}.
  Lemma~\ref{lemma1} shows the possible states in which the receiver can be after reading the 
  synchronization sequence $\BSS$. Lemma~\ref{lemma2} extends this result to show that 
  for all these possible states bytes can be sampled correctly. 
  Theorem~\ref{transcorr} -- the main correctness theorem of the hardware interface -- is proven 
  by induction on the number of bytes in a message. We conclude this Section by showing
  how Lemma~\ref{lemma2} is used in the induction step. 

  \subsection{Assumptions: sender correctness}

  The sender is proven to effectively generate each bit for eight clock cycles. 
  This discharges the digital hypotheses of Proposition \ref{thm:dig-thm}. 
  Formally, this is defined as follows, where $l$ denotes the number of bytes in a message:
  \begin{definition} {\bf Correctness of $\ce_s$.}
    \begin{displaymath}
      \begin{array}{lcl}
        \WF_\ce(\ce_s, l, k, c) & \equiv & \forall i < l, \ce_s[c + 8 \cdot i] = 0 
        \wedge 
        \forall j \in [1:k], \ce_s[c + 8 \cdot i + j] = 0 
      \end{array}
    \end{displaymath}  
  \end{definition}

  We prove that the sender generates frames with the specified format.
  For the purpose of this paper, we are only concerned with synchronization bits, i.e., the $\BSS$ sequence. 
  This is expressed by the following predicate, where $l$ denotes the number of bytes in a message:
  \begin{definition} {\bf Partial Correctness of $\In_s$.}
    \begin{displaymath}
      \begin{array}{ll}
        \WF_\In(\In_s,l,c) \equiv & \forall i<l, \forall y \in [0:7] 
        \left\{ 
          \begin{array}{cl}
            & \In_s[c + 80 * i + 16 + y - 1] = 1 \\ 
            \wedge & \In_s[c + 80 * i + 24 + y - 1] = 0          
          \end{array} \right.
      \end{array}
    \end{displaymath} 
  \end{definition}

 \subsection{Proof roadmap}
 \label{sec:proof-roadmap}

  Senders send frames composed by two initial bit sequences -- the transmission start sequence ($\TSS$) and the frame start sequence ($\FSS$) --
  followed by a number of bytes, say $l$ bytes.
  Before sending each byte a synchronization sequence -- the falling edge made of $\BSS[0]$ and $\BSS[1]$ --
  is sent. Our objective is to prove that each byte is received properly. We take as reference 
  the time point when a receiver reads the first copy of bit $\BSS[0]$.
  In this section, this time point is referred to with $t$.
  When reading the initial bit sequences or the previous byte, receivers have different configurations when starting reading a byte 
  at time $t$. 
  We first show the different possible states after reading the initial bit sequences, i.e., the states before reading the first byte.

  \begin{figure}[ht]
    \centering
    \input{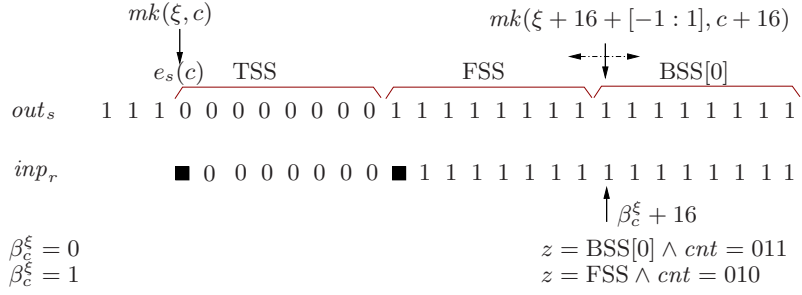}
    \caption{Initial Transmission Phase}
    \label{fig:transmit}
  \end{figure}

  We illustrate these different possible states in Figure~\ref{fig:transmit}. 
  The first two lines show the output of the sender and how it is seen by the receiver.
  Black boxes indicate possible metastabilities.
  We first need a mark and assume that cycle $\xi$ is the first affected cycle. 
  Because of clock drift, the $\BSS[0]$-mark may appear on the receiver side
  15 ($\chi = -1 $), 16 ($\chi = 0 $) or 17 ($\chi = 1 $) cycles after $\xi$. There is a potential metastability at cycle $\xi$.
  Depending on the value reached after resolution --~that is depending on the value of 
  $\beta_c^\xi$~-- the receiver automaton reaches different state and counter values
  when the $\BSS[0]$-mark is detected. In the figure, we show these values at $\beta_c^\xi + 16$, where 
  the automaton is either in state $\BSS[0]$ with a counter at $011$ or in state $\FSS$ with a counter at $010$.
  This corresponds to the case where $\chi = 0$. If $\chi = 1$, then the receiver automaton reaches 
  state $\BSS[0]$ 17 cycles after $\xi$ and with counter value $100$.


  The base case and all the sub-cases of the induction step are proven
  using two main lemmas. The first one states that synchronization 
  occurs while sampling the synchronization sequence. It shows for all possible
  states of the receiver before reading the $\BSS$ sequence which are 
  the possible states after reading this sequence.
  The second 
  lemma shows that this synchronization is good enough to sample
  a byte. It shows for all possible states after reading the $\BSS$ sequence
  the byte can be sampled properly.
  Its proof shows the correct value of the counter used
  in the bit-clock synchronization algorithm.

\subsection{Main statement}
\label{sec:proof:main}

 The main theorem is shown below. The first line contains hypotheses about 
 the low level aspects ($\H$), our integration of the analog and the digital worlds (Equation~\ref{eq:mix-ad}),
 and our assumption that the sender is correct. 
 The first bit of the message is 
 put on the bus as sender cycle $c$ and the first ``affected" receiver cycle is cycle $\xi$.
 These two facts are hidden in hypothesis $\H[k \triangleleft 7,n \triangleleft 6]$.
 In the second line, we simply assume that the initial state is 
 $\mathit{idle}$ and that each bit of a byte is indexed by $j$.

 \begin{theorem} \textbf{Transmission Correctness.}
   \label{transcorr}
   \begin{displaymath}
     \begin{array}{ccl}
       & &   \H[k \triangleleft 7,n \triangleleft 6]  
       \wedge \mbox{Equation ~\ref{eq:mix-ad}} \wedge \WF_\ce(\ce_s, l, k, c) \wedge \WF_\In(\out_s,l,c) \\
       & \wedge &  z^\xi = \mathit{idle} \wedge j \in [7:0]\\
       \rightarrow & & \\
       \forall i < l,& \exists \nu, &
                        \cy(\nu, c + 16 + 80 \cdot i) \mbox{ (* bss0 mark *)} \\ 
        \wedge\\ &        & \cy(\nu + 7, c + 24 + 80 \cdot i) \mbox{(*bss1 mark*)}\\
               &        & z^{\nu + 78} = b[7] \wedge \cnt^{\nu + 78} = 010 
                          \wedge \BYTE^{\nu + 79} = \tup{\out_s^{c + 16 + 80 \cdot i + 8 \cdot (j + 2)}} \\ 
               &  \vee  & \\
               &        & (\cy(\nu + 7, c + 24 + 80 \cdot i) \vee \cy(\nu + 8, c + 24 + 80 \cdot i)) \mbox{(*bss1 mark*)}\\
               &        & z^{\nu + 79} = b[7] \wedge \cnt^{\nu + 79} = 010 
                          \wedge \BYTE^{\nu + 80} = \tup{\out_s^{c + 16 + 80 \cdot i + 8 \cdot (j + 2)}} \\ 
               &  \vee  & \\
               &        & (\cy(\nu + 8, c + 24 + 80 \cdot i) \vee \cy(\nu + 9, c + 24 + 80 \cdot i)) \mbox{(*bss1 mark*)}\\
               &        & z^{\nu + 80} = b[7] \wedge \cnt^{\nu + 80} = 010 
                          \wedge \BYTE^{\nu + 81} = \tup{\out_s^{c + 16 + 80 \cdot i + 8 \cdot (j + 2)}} \\ 
               &  \vee  & \\
               &        & \cy(\nu + 9, c + 24 + 80 \cdot i) \mbox{(*bss1 mark*)}\\
               &        & z^{\nu + 81} = b[7] \wedge \cnt^{\nu + 81} = 010 
                          \wedge \BYTE^{\nu + 82} = \tup{\out_s^{c + 16 + 80 \cdot i + 8 \cdot (j + 2)}} \\ 
      \end{array}
    \end{displaymath}
  \end{theorem}

 The conclusion reads "for each byte $i$ of a message containing $l$ bytes,
 there exists a receiver cycle $\nu$ from which the receiver samples the byte correctly".
 The conclusion is a conjunction of two terms. The first one expresses the knowledge 
 of the receiver about the first bit of the synchronization sequence ($\BSS[0]$).
 The knowledge of this "mark" and clock drift induces three possible positions for the second
 bit of this sequence. These three possible positions imply four different 
 ways of sampling a byte. This is explained using the drift factor $\chi$ and the 
 metastability factor $\beta$. The former can take three values: $-1$, $0$, or $1$.
 The latter is either $0$ or $1$. Combining these two factors, we obtain four
 possibilities: $-1,0,1,2$. Because of these four possibilities, the total number of 
 cycles needed to sample a byte can have four different values.
 These values are expressed by the second term of the conclusion.
 In each case, state and counter values mean that the last bit has been sampled.
 One cycle after that, the byte register is properly updated.
 Finally, it takes between 79 to 82 cycles to sample a byte.

 This theorem also proves lower and upper bounds on the time at which the last bit of one byte is recovered.
 From a simple computation based on the marks of the conclusion, 
 these bounds can be expressed as functions of the reference clock ($\tau_\tref$) and the time ($e_s(c)$) when
 the first bit is put on the bus by the sender.

\subsection{Lemma 1: Crossing synchronization edges}
\label{sec:lemma1}

  The objective of Lemma~\ref{lemma1} is to prove for all possible states 
  before reading the synchronization sequence ($\BSS$) what are the 
  possible states after reading this sequence. 
  We illustrate the proof when reading the first byte, i.e., 
  after reading the initial bit sequences. 
  The other cases are discussed in Section~\ref{sec:in-step} about the induction step. 
  Our reasoning is illustrated in Figure~\ref{fig:sync-edge}.
  The first two lines show the output of the sender and how it is seen by the receiver.
  Black boxes indicate possible metastability.
  Question marks are used to denote unknown values. 

\begin{figure}
  \centering
  \input{sync_edge.pstex_t}
  \caption{Traversing synchronization edges}
  \label{fig:sync-edge}
\end{figure}

  As explained in Section~\ref{sec:proof-roadmap}, the receiver may be in different configurations at the time
  of the detection of the $\BSS[0]$-mark.
  We fix the initial step of the lemma to match the time of the detection of the $\BSS[0]$-mark. 
  We consider the case where the receiver is in state $\BSS[0]$ with a counter value at either $011$ or $100$. 
  The other cases are proven in a similar way.
  
  According to Proposition~\ref{thm:reg-affw} and assuming that the $\BSS[0]$-mark is known, 
  the $\BSS[1]$-mark has three possible times (one for each value of $\chi$).
  The potential metastability around that edge has the same three times.
  We consider bits sampled by the receiver at these times unknown. 
  At most, 
  three bits are unknown.    
  Depending on the values of these three bits, the automaton will spend more or less time in the states of $\BSS$.
  There is synchronization if the lower and the upper bound on this number of cycles allow proper sampling.
  These bounds are defined by Lemma~\ref{lemma1}, which proves that the automaton reaches 
  state $b[0]$ with counter value $011$ in at least 15 and at most 18 cycles.
  We now explain the proof of these lower and upper bounds.

  Let $t$ be the time of the affected cycle of $\BSS[0]$.
  If the three unknown bits are 0 (see line 3 "earlier sync" in Figure~\ref{fig:sync-edge}), 
  signal $\sync$ is high at $t + 7 + 4 = t + 11$. The counter is reset and signal $\strobe$ 
  is high at $t + 11 + 3 = t + 14$. In the next cycle, the automaton reaches state $z^{t + 15} = b[0]$.
  For any lower value of the counter, the automaton will reach this state earlier.

  If the unknown bits are 1 (see line 4 "latest sync" in Figure~\ref{fig:sync-edge}), 
  signal $\sync$ is high at $t + 10 + 4 = t + 14$. If the counter was $100$ initially, then it has reached
  value $010$ and $\strobe$ is high. 
  At the same time, signal $\sync$ is high, the automaton stays in $\BSS[0]$, and the counter is reset.
  At cycle $t + 17$, $\strobe$ is high and the automaton reaches $b[0]$
  with a correct counter value at $t+18$.
  For any larger value, the automaton requires more cycles to reach this state.

  The main statement includes all possible configurations when sampling the first byte, but 
  also all possible configurations when sampling byte $i$ in the induction step. 
  This statement assumes that the time of the first bit of the first part of the 
  synchronization sequence (the "$\BSS[0]$-mark") is known by the reader ($\cy(t, c + 16)$). 
  The conclusion contains the four possible ways to sample the synchronization sequence.
  In each term of the disjunction, we know the position of the first bit of the 
  second part of the synchronization sequence ($\BSS[1]$). This knowledge is crucial 
  to prove that synchronization is good enough to sample bits correctly (Lemma~\ref{lemma2}, Section~\ref{sec:lemma2}).

\begin{lemma}{\bf Synchronization}
  \label{lemma1}
  \begin{displaymath}
    \begin{array}{cl}
             & \cy(t, c + 16) \mbox{(* bss0 mark *)}\\
      \wedge & \mbox{(*Next are hyps. about starting point for the lemma*)} \\
             & \; (z^t = \BSS[0] \wedge \cnt^t \in \{ 011, 100 \} 
               \vee z^t = \FSS[0] \wedge \cnt^t \in \{ 001 , 010 \} \\
            & \vee   z^t = b[7] \wedge \cnt^t \in \{ 001, 010, 011 ,100\}) \mbox{(*for induction step*)}\\
      \rightarrow \\
             & \cy(t + 7, c + 24) \wedge \mbox{(* bss1 mark *)}\\
             & z^{t+15} = b[0] \wedge \cnt^{t + 15} = 011\\
      \vee   \\
             & (\cy(t + 7, c + 24) \vee \cy(t + 8, c + 24)) \wedge \mbox{(* bss1 mark *)}\\
             & z^{t+16} = b[0] \wedge \cnt^{t + 16} = 011\\
      \vee   \\
             & (\cy(t + 8, c + 24) \vee \cy(t + 9, c + 24)) \wedge \mbox{(* bss1 mark *)}\\
             & z^{t+17} = b[0] \wedge \cnt^{t + 17} = 011 \\
      \vee   \\
             & \cy(t + 9, c + 24) \wedge \mbox{(* bss1 mark *)}\\
             & z^{t+18} = b[0] \wedge \cnt^{t + 18} = 011
    \end{array}
  \end{displaymath}
\end{lemma}

  This lemma is proven following the same approach as the one used to prove the correctness 
  of voted bits (see Section~\ref{ex:voted-bits}).
  The idea is to obtain from the low level model which inputs are unknown 
  and prove that the design works properly for all possible values of 
  these unknown inputs.

  As suggested by the informal description of the proof, there are at most three 
  unknown inputs. In fact, for each position of the $\BSS[1]$-mark, only one 
  input is unknown. It is the bit that appears exactly on this mark and it is 
  due to metastability. 
  For instance, if the mark is at its earlier position 
  (see line ``Earliest sync'' in Figure~\ref{fig:sync-edge}), then the only 
  unknown input is at time $t + 7$. Because instantiating Proposition~\ref{thm:dig-thm} 
  for this mark gives us that bits from $t + 8$ are known.
  In a similar way, if the mark is at its latest position (see line 4 "latest sync" in Figure~\ref{fig:sync-edge}), 
  then we also know that inputs at $t + 7$ and $t + 8$ must be one. 
  This is due to our assumption that clock drift is bounded, i.e., bits at $t + 7$ and $t + 8$ 
  are still in the safe sampling window starting at the $\BSS[0]$-mark.

  For each position of the $\BSS[1]$-mark, we prove a lemma on the digital design, which 
  shows the two possible lengths to sample the synchronization sequence.
  For instance, if the $\BSS[1]$-mark appears at $t + 7$, we prove that sampling the 
  synchronization sequence takes 15 to 16 cycles.
  This is expressed in Proposition~\ref{prop:bss-seq} below. 
  The hypotheses first state that at times $t+1$ until $t + 6$ six good copies 
  of $\BSS[0]$ are known and that at times $t + 8$  until $t + 13$ six good copies 
  of $\BSS[1]$ are known. Input at time $t + 7$ is left unspecified and the model-checker
  will have to consider all possible values. The rest of the hypotheses state 
  the different receiver states and counter-value that are possible at time $t$ when reading 
  the first copy of $\BSS[0]$. The conclusion shows the time points when the receiver 
  reaches state $b[0]$ with counter-value $011$, i.e., after reading the $\BSS$ sequence. 

\begin{proposition}
  \label{prop:bss-seq}
  \begin{displaymath}
    \begin{array}{cl}
             & \forall u \in [0:5], inp^{t + 1 + u} = 1\ 
      \wedge \forall v \in [0:5], inp^{t + 8 + v} = 0 \ \mbox{(* 2x6 bits known *)}\\ 
     \wedge  & \; (z^t = \BSS[0] \wedge \cnt^t \in \{ 010 , 011\}\ 
                \vee z^t = \FSS \wedge \cnt^t \in \{ 001, 010 \}  \mbox{(*a*)} \\
             & \vee z^t = b[7] \wedge \cnt^t \in \{ 001 , 010 , 011 , 100\})  \mbox{(*b*)}\\
   \rightarrow & \\
         & z^{t + 15} = b[0] \wedge \cnt^{t + 15} = 011
   \vee    z^{t + 16} = b[0] \wedge \cnt^{t + 16} = 011
    \end{array}
   \end{displaymath}  
\end{proposition}
\begin{proof}
  By IHaVeIt. For efficiency of the computations, we decompose this proposition in two propositions. 
  We prove the conclusion assuming hypotheses $\mbox{(*a*)}$ in one proposition. 
  In another one, we prove the conclusion assuming \mbox{(*b*)}. 
  Each proof is fully automatic. 
\end{proof}

\subsection{Lemma 2: Sampling bytes correctly}
\label{sec:lemma2}

\begin{figure}
  \centering
  \input{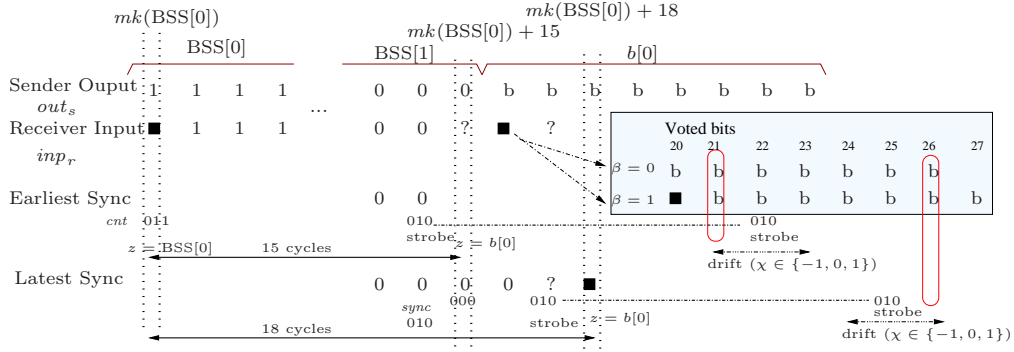}
  \caption{Sampling bytes correctly}
  \label{fig:sampling}
\end{figure}

  The previous Lemma shows the different possible states of the receiver 
  after and before reading the $\BSS$ sequence. Lemma~\ref{lemma2} shows 
  that it is possible to sample the transmitted byte for all these possibilities.
  Let $t$ denote the receiver cycle reading the first copy of $\BSS[0]$.
  To simplify, we only consider the case where the receiver is in state 
  $z^t = \BSS[0]$ with counter $\cnt^t = 011$ when reading the first bit 
  of $\BSS[0]$. This case is pictured in Figure~\ref{fig:sampling}. All other cases would be proven 
  in a similar way.
  The first two lines show the digital output ($\out_s$ in Figure~\ref{fig:mixed-scheme}) of the sender and the digital 
  input of the receiver ($\inp_r$ in Figure~\ref{fig:mixed-scheme}). Black boxes show potential
  metastabilities and a '?' illustrates the fact that the eighth copy is not certainly correct.   
  Line "Earliest Sync" considers the shortest traversal of the synchronization sequence.
  State $b[0]$ is reached after 15 cycles ($t + 15$). Line "Latest Sync" considers the longest
  traversal of the synchronization sequence. State $b[0]$ is reached after 18 cycles ($t + 18$).
  The large box shows the values of the voted bit, which is simply the receiver input 
  delayed by four cycles. The first line shows the case of no metastability -- or good resolution of it ($\beta = 0$) -- when reading the 
  first copy of $b[0]$. The second lines shows bad resolution of the metastability ($\beta = 1$). 
  Numbers indicate cycle numbers counting from time $t$ when the receiver reads the 
  first copy of $\BSS[0]$.
  The possible strobe points for the earliest and latest synchronization are also shown. 
  In total there are four possible strobe points depending on the four durations of traversing the synchronization sequence.
  Strobing appears at $t + [15:18] + 7$ for $b[0]$. In general, for any bit $b[j]$ strobing appears
  at $t + [15:18] + 8 \cdot j + 7$.
  We see here that the position of the strobe points is fully 
  determined by the traversal of the synchronization sequence $\BSS$. It does not depend on clock drift.
  In contrast, the relative position of the voted bits may shift by one cycle. This is represented 
  by the different values of drift factor $\chi$.
  The objective of Lemma~\ref{lemma2} is to prove that all strobe points coincide with 
  a good voted bit.
  
  \subsubsection{Positions of the voted bits}

  We assumed that at receiver cycle $t$ the first bit of $\BSS[0]$ is read.
  This bit was put on the bus at sender cycle $c$. 
  We are interesting in the bits sent after the two bits of the $\BSS$ sequence.
  We want to read bit $b[j]$ which is put on the bus $8 \cdot (j + 2)$ cycles after cycle $c$. 
  We simply instantiate Proposition~\ref{prop:voted-bit2} (Section~\ref{ex:voted-bits})
  with $\xi = t$ and $\alpha = 8 \cdot (j + 2)$
  and obtain the following:
  \begin{displaymath}
    \bigvee_{\chi \in \{-1, 0 , 1\}} \forall x \in [4:10], v^{t + 8 \cdot (j + 2) + \chi + \beta_{c + 8 \cdot (j + 2)}^{t + 8 \cdot (j + 2)+ \chi} + x} = \out_s[c+8 \cdot (j + 2) -1]    
  \end{displaymath}
  The right hand side of this equation denotes the value of the bit sent by the sender unit, i.e., $b[j]$.
  The left hand-side shows the different positions of the voted bits. 
  There are seven good bits, i.e., one for every value of $x \in [4:10]$.
  Ideally, each bit $b[j]$ is read at receiver cycle $t + 8 \cdot (j + 2)$.
  Because of clock drift, this can suffer an error of one cycle.
  Each bit $b[j]$ is then read at cycle $t + 8 \cdot (j + 2) + \chi$.
  For each bit $b[j]$, there might be a metastability when reading the first copy of it, i.e., at receiver
  cycle $t + 8 \cdot (j + 2) + \chi$. The effect of this metastability is expressed by $\beta_{c + 8 \cdot (j + 2)}^{t + 8 \cdot (j + 2)+ \chi}$, 
  simply written $\beta$ in the remainder of this section.

  For every bit $b[j]$, the position of the corresponding good voted bits is equal to the following expression:
  \begin{displaymath}
    t + 8 \cdot (j + 2) + \beta + \chi + x
  \end{displaymath}

  The objective of Lemma~\ref{lemma2} is to show that there is always an $x \in [4:10]$ to match
  the position of a strobe point with a voted bit. Formally, we have to solve the following equality 
  where the left hand side corresponds to strobe points and the right hand side to 
  the cycles at which the voted bit is correct.
  \begin{equation}
    t + [15:18] + 8 \cdot j + 7 = t + \alpha + \beta + \chi + x 
  \end{equation}




  \subsubsection{Smallest value of $x$}

  The minimum $x$ is required when the right hand side is maximized and the left hand side of the equality is the earliest cycle.
  This means that the receiver is one cycle behind the sender. Because clock ticks differ at most by one,
  this implies that $\chi$ cannot take value 1.
  The right hand side is therefore maximized with
  $\beta = 1$ and $\chi = 0$.
  We need to find $x$ such that:
  \begin{displaymath}
    t + 15 + 8 \cdot j + 7 = t + 16 + 8 \cdot j + 1 + 0 + x 
  \end{displaymath}
  The solution is $x = 5$. 
  If the receiver would strobe at counter value $001$, traversing the synchronization edges would take 
  a cycle less. Strobe points would be positioned at $t + [14:17] + 8 \cdot j + 7$. 
  The above equation would become $t + 14+ 8 \cdot j + 7 = t + 16 + 8 \cdot j + 1 + 0 + x $.
  The solution would be $x = 4$ and would still be in the interval $[4:10]$.
  This means that counter value $001$ would 
  be a limit, i.e., the earliest working synchronization point. 
  \begin{statement}
    \label{statement1}
    The lowest reset value of counter $\cnt$ is $001$. As the counter is reset to $000$, the lowest difference
    between the strobe and the reset points is one cycle.
  \end{statement}

  \subsubsection{Largest value of $x$}
  
  The maximum $x$ is required when the right hand side is minimized and the left hand side of the equality is the latest cycle.
  This means that the receiver is one cycle ahead of the sender. Again, because of the bound on clock drift,
  this implies that $\chi \neq -1$.
  The right hand side is therefore minimized with $\beta = 0$ and $\chi = 0$.
  Here, we need to find $x$ such that:
  \begin{displaymath}
   t+ 18 + 8 \cdot j + 7 = t+ 16 + 8 \cdot j + 0 + 0 + x 
  \end{displaymath}
  The solution is $x = 9$. 
  If the receiver would strobe at counter value $011$, traversing the synchronization sequence would take
  one cycle more. Strobe points would be positioned at $t + [16:19] + 8 \cdot j + 7$.
  The above equality would become $t + 19+ 8 \cdot j + 7 = t + 16 + 8 \cdot j + 1 + 0 + x$.
  The solution would be $x = 10$ and would still be in the interval $[4:10]$.
  Note that this counter value is equivalent to the one proposed by the Flex Ray standard~\cite{PS05}. 
  Value $100$ proposed in~\cite{AutoICCD05} would be outside this limit and is therefore not adequate.
  \begin{statement}
    \label{statement2}
    The largest reset value of counter $\cnt$ is $011$. As the counter is reset to $000$, the largest
    difference between the strobe and the reset points is three cycles. 
  \end{statement}

  \subsubsection{Statement of Lemma~\ref{lemma2}}

  The main statement builds on Lemma~\ref{lemma1} and shows the four different ways of sampling 
  a byte, starting from the first bit of the synchronization sequence. 
  In each way, we also have the knowledge of the $\BSS[1]$-mark. 

\begin{lemma}
  \label{lemma2}
  {\bf Sampling bytes correctly.}
  \begin{displaymath}
    \begin{array}{cl}
                & \cy(t, c + 16) \\
     \wedge  & \; (z^t = \BSS[0] \wedge \cnt^t \in \{ 010 , 011\}\ 
                \vee z^t = \FSS \wedge \cnt^t \in \{ 001 , 010 \}    \\
             & \vee z^t = b[7] \wedge \cnt^t \in \{ 001 , 010 , 011 ,100 \}) \\
     \rightarrow \\
             & \cy(t + 7, c + 24) \\
             & \wedge z^{t+78} = b[7] \wedge \cnt^{t + 78} = 010 
               \wedge \BYTE^{t + 79} = \tup{\out_s^{c + 16 + 8 \cdot (j + 2)}}, j \in [7:0] \\
       \vee   & \\
              & (\cy(t + 7, c + 24) \vee \cy(t + 8, c + 24)) \\
             & \wedge z^{t+79} = b[7] \wedge \cnt^{t + 79} = 010 
               \wedge \BYTE^{t + 80} = \tup{\out_s^{c + 16 + 8 \cdot (j + 2)}}, j \in [7:0] \\
       \vee   & \\
              & (\cy(t + 8, c + 24) \vee \cy(t + 9, c + 24)) \\
             & \wedge z^{t+80} = b[7] \wedge \cnt^{t + 80} = 010 
               \wedge \BYTE^{t + 81} = \tup{\out_s^{c + 16 + 8 \cdot (j + 2)}}, j \in [7:0] \\
       \vee   & \\
              & \cy(t + 9, c + 24) \\
             & \wedge z^{t+81} = b[7] \wedge \cnt^{t + 81} = 010 
               \wedge \BYTE^{t + 82} = \tup{\out_s^{c + 16 + 8 \cdot (j + 2)}}, j \in [7:0] \\

    \end{array}
  \end{displaymath}
\end{lemma}

  \subsection{Induction Step}
  \label{sec:in-step}

    We perform a proof by induction on the number of bytes that is transmitted. 
    We first extract an arbitrary cycle $\nu$ from our induction hypothesis.
    We know that from this cycle we can sample byte $i$ properly.
    We then use Lemma~\ref{lemma2} to find a cycle from which 
    byte $i + 1$ can also be recovered. This shows the existential quantification. 
    For byte $i$, the induction hypothesis gives us a $\BSS[0]$-mark.
    The other part of the induction hypothesis gives us four possible
    completion times for sampling byte $i$.
    Our induction hypothesis for an arbitrary $\nu$ is as follows:
     \begin{displaymath}
     \begin{array}{ccl}
       & &
                        \cy(\nu, c + 16 + 80 \cdot i) \mbox{ (* bss0 mark *)} \\ 
        \wedge\\ &        & \cy(\nu + 7, c + 24 + 80 \cdot i) \mbox{(*bss1 mark*)}\\
               &        & z^{\nu + 78} = b[7] \wedge \cnt^{\nu + 78} = 010 
                          \wedge \BYTE^{\nu + 79} = \tup{\out_s^{c + 16 + 80 \cdot i + 8 \cdot (j + 2)}} \\ 
               &  \vee  & \\
               &        & (\cy(\nu + 7, c + 24 + 80 \cdot i) \vee \cy(\nu + 8, c + 24 + 80 \cdot i)) \mbox{(*bss1 mark*)}\\
               &        & z^{\nu + 79} = b[7] \wedge \cnt^{\nu + 79} = 010 
                          \wedge \BYTE^{\nu + 80} = \tup{\out_s^{c + 16 + 80 \cdot i + 8 \cdot (j + 2)}} \\ 
               &  \vee  & \\
               &        & (\cy(\nu + 8, c + 24 + 80 \cdot i) \vee \cy(\nu + 9, c + 24 + 80 \cdot i)) \mbox{(*bss1 mark*)}\\
               &        & z^{\nu + 80} = b[7] \wedge \cnt^{\nu + 80} = 010 
                          \wedge \BYTE^{\nu + 81} = \tup{\out_s^{c + 16 + 80 \cdot i + 8 \cdot (j + 2)}} \\ 
               &  \vee  & \\
               &        & \cy(\nu + 9, c + 24 + 80 \cdot i) \mbox{(*bss1 mark*)}\\
               &        & z^{\nu + 81} = b[7] \wedge \cnt^{\nu + 81} = 010 
                          \wedge \BYTE^{\nu + 82} = \tup{\out_s^{c + 16 + 80 \cdot i + 8 \cdot (j + 2)}} \\ 
      \end{array}
    \end{displaymath}

    This induction hypothesis is pictured in Figure~\ref{fig:ind-step}.
    It starts with our arbitrary cycle $\nu$ at which the receiver starts
    sampling byte $i$. Our induction hypothesis contains also
    marks (cycle $\mu$) for the second part of the $\BSS$ sequence.
    Sampling byte $i + 1$ starts at cycle $\nu'$ and $\BSS[1]$ is seen
    at cycle $\mu'$. The idea of the proof is (1) knowing $\nu$ and $\mu$
    we obtain values for $\nu'$ and $\mu'$ using Proposition~\ref{thm:reg-affw}
    and (2) we instantiate Lemma~\ref{lemma2} at $\nu'$ to show
    that byte $i + 1$ can be sampled properly.

    \begin{figure}
      \centering
      \input{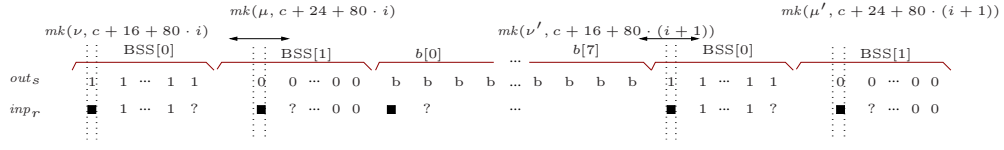}
      \caption{Induction step}
      \label{fig:ind-step}
    \end{figure}

    While sampling byte $i$, clocks are likely to drift. Consequently, 
    the $\BSS[0]$-mark (i.e., position of $\nu'$) for the next byte is known with the error defined by 
    $\chi$. From Proposition~\ref{thm:reg-affw} with $\xi = \nu$ and $\alpha = 80 \cdot i$, 
    we have $\nu' = \nu + 80 + \chi$. Thus,
    we have three different possible $\BSS[0]$-marks for byte $i + 1$: 
    \begin{displaymath}
        \bigvee_{\chi \in \{-1, 0, 1\}} \cy(\nu + 80 + \chi, c +16 + 80 \cdot (i + 1))
    \end{displaymath}

    To use Lemma~\ref{lemma2}, we also need to know when the receiver control 
    automaton sampled the last bit of the byte, i.e., when $z = b[7]$. 
    There are four possible times corresponding to the four possible
    completion times of sampling byte $i$. At the end, we have 
    $4  \ast 3 = 12$ cases in our induction step.

    Let us consider the case where there is no drift
    and the receiver sees the first bit of the 
    next $\BSS$ sequence exactly 80 cycles after the 
    previous one. Formally, we have: 
     \begin{displaymath}
        \cy(\nu + 80, c +16 + 80 \cdot (i + 1)) 
    \end{displaymath}
    
    For this perfect "mark" there are still four possible 
    ways for the receiver to sample byte $i$. From the induction 
    hypothesis, there
    are four possible times when the 
    receiver is sampling the last bit and ready to strobe and move to state $\BSS[0]$,
    i.e., the receiver is in state $b[7]$ with $\cnt = 010$. Formally, we have the following
    cases:
    \begin{equation}
      \label{4-cases}
      \begin{array}{cl}
                    & z^{\nu + 78} =  b[7]  \wedge  \cnt^{\nu + 78} =  010 \\
          \vee   & z^{\nu + 79} =  b[7]  \wedge  \cnt^{\nu + 79} =  010 \\
          \vee   & z^{\nu + 80} =  b[7]  \wedge  \cnt^{\nu + 80} =  010 \\
          \vee   & z^{\nu + 81} =  b[7]  \wedge  \cnt^{\nu + 81} =  010 
      \end{array}
    \end{equation}
    
    To instantiate Lemma~\ref{lemma2}, we need to obtain that the receiver is in state $b[7]$ with a proper
    counter value when the sender puts the first bit of  $\BSS$ on the bus, i.e., at $\nu + 80$, 
    the time of the $\BSS[0]$-mark.
    In the earliest case, the counter reaches value $010$ at $\nu + 78$. So, at time $\nu + 80$ the counter has value $100$.
    In the latest case, the counter equals $010$ at $\nu + 81$ and therefore it has value $001$ at time $\nu + 80$.
    In the remaining two cases, the values at time $\nu + 80$ would be $011$ or $010$.
    For all these cases, the premises of Lemma~\ref{lemma2} are satisfied for $t = \nu + 80$ and $c = c + 80 \cdot i$.
    This shows that there exist four possible cycles at which byte $i+1$ can be sampled
    properly for each one of these times.

    Let us consider the case when the $\BSS[0]$-mark is early.
    Formally, we have:
     \begin{displaymath}
        \cy(\nu + 79, c +16 + 80 \cdot (i + 1)) 
    \end{displaymath}
    
    We have the same cases as Equation~\ref{4-cases}. 
    But we must consider the counter value at time $t = \nu + 79$ instead of $\nu + 80$.
    Assume the latest
    sampling time (the fourth case in Equation~\ref{4-cases}). 
    The counter equals $010$ and $z = b[7]$ at time 
    $\nu + 81$. Consequently, at time $\nu + 79$ the counter is $000$
    and the receiver automaton is still sampling $b[7]$.     
    Under this configuration it is not possible to instantiate
    Lemma~\ref{lemma2} and in fact the receiver 
    would not be able to synchronize. 
    Note that the induction hypothesis also gives us the time of the
    $\BSS[1]$-mark for byte $i$. For this latest case, we know that this mark
    was seen by the receiver at time $\nu + 9$, i.e., we have $\cy(\nu + 9, c + 24 + 80 \cdot i)$.
    Applying Proposition~\ref{thm:reg-affw} on this $\BSS[1]$-mark with $\alpha = 72$ gives us 
    another three possible times for the $\BSS[0]$-mark of byte $i + 1$.
    Formally, we have:
    \begin{equation}
      \label{eq:3}
       \bigvee_{\chi \in \{-1, 0 , 1\}}  \cy(\nu + 9 + 72 + \chi, c + 24 + 80 \cdot i + 80)
    \end{equation}
    This means that the $\BSS[0]$-mark is at the earliest ($\chi = -1$) at $ \nu + 80$.
    This contradicts our assumption that the mark is at $\nu + 79$.
    The remaining cases when the $\BSS[0]$-mark is at $\nu + 79$ are proven 
    using Lemma~\ref{lemma2} as explained above.

    Let us consider the case when the $\BSS[0]$-mark is late.
     Formally, we have:
     \begin{displaymath}
        \cy(\nu + 81, c +16 + 80 \cdot (i + 1)) 
    \end{displaymath}
    Again we have the same cases as Equation~\ref{4-cases} 
    and must consider the counter value at time $\nu + 81$.
    Assume the earliest sampling time of byte $i$, i.e.,
    the counter has value $010$ at time $\nu + 78$.
    This means that at time $\nu + 81$, counter has 
    value $101$ and the receiver is in state $\BSS[0]$.
    Under this configuration it is not possible to use 
    Lemma~\ref{lemma2}. Here again we use the 
    fact that the induction hypothesis gives us a
    time for the $\BSS[1]$-mark from which we can 
    derive a contradiction. In the earliest sampling time,
    this mark is for byte $i$ at time $\nu + 7$ and we have 
    $\cy(\nu + 7, c + 24 + 80 \cdot i)$. We apply 
    Proposition~\ref{thm:reg-affw} on this $\BSS[1]$-mark with $\alpha = 72$ to obtain
    three possible times for the $\BSS[0]$-mark of byte $i + 1$.
    Formally, we have:
    \begin{equation}
      \label{eq:4}
      \bigvee_{\chi \in \{-1, 0 , 1\}}  \cy(\nu + 7 + 72 + \chi, c + 24 + 80 \cdot i + 80)
    \end{equation}
    This means that the $\BSS[0]$-mark is at the latest ($\chi = +1$) at $ \nu + 80$.
    This contradicts our assumption that the mark is at $\nu + 81$.
    The remaining cases when the $\BSS[0]$-mark is at $\nu + 81$ are proven 
    using Lemma~\ref{lemma2} as explained earlier in this section.

\section{Related work}
\label{sec:related}

  The first verification effort about physical layer protocols was carried out by Moore~\cite{BIPHIM93}.
  Moore developed a general model of asynchronous communications as a function 
  in the logic of the ACL2 theorem prover~\cite{KM00}. Moore's model assumes distortion 
  around sampling edges and does not allow for clock jitter. 
  Sender and receiver modules 
  are also represented by two functions. Moore's correctness criterion states that 
  the composition of these three functions is an identity. He applied this approach 
  to the verification of a Biphase-Mark protocol. 
  
  Moore's work inspired many studies around this protocol.
  Recently, Vaandrager and de Groot~\cite{vaandrager06} modeled the protocol 
  and analog behaviors using a network of timed-automata. 
  Their model is slightly more general than Moore's and allows for clock jitter.
  They can derive tighter bounds for the Biphase-Mark protocol.
  Previously, timed-automata have been used to verify a low level protocol
  based on Manchester encoding and developed by Philips~\cite{bosscher94verification}.
  Another recent proof of the Biphase-Mark protocol has been proposed by Brown and Pike~\cite{brown_pike_06}.
  They developed a general model of asynchronous communications in the formalism
  of the tool SAL~\cite{SAL2} developed at SRI. 
  Their model includes clock jitter and metastability.
  Using $k$-induction, the verification of the parameterized specification of Brown and Pike is largely automatic.
  All these studies tackle \textit{protocol specification} only and not actual hardware implementation.
  They prove functional correctness. 
  We prove a more precise theorem about a gate-level hardware implementation and from which bounds 
  on the transmission duration can be derived. 

  The verification of analog and mixed signal (AMS) designs is a relatively young research field. A recent survey
  gave an overview of this emerging research area~\cite{AMSVerif08}. 
  The authors identify several successful applications of  automatic techniques (equivalence checking, model checking, or run-time verification) 
  in the context of AMS designs.  Our work is more related to the last category identified in this survey, namely proof based methods.
  Hanna~\cite{hanna94reasoning,hanna98} used predicates to approximate analog behaviors at the transistor level.
  The predicates can be embedded in digital proofs. 
  His work is not specifically targeted to communication circuits and does not consider timing parameters, metastability or clock drift.
  We consider only gates and not their structure in terms of transistors.
  Recently, Al Sammane \textit{et al.}~\cite{Al-SammaneZT07} proposed a new symbolic verification methodology based on the computer algebra system Mathematica.
  This approach is based on  a combination of induction and symbolic simulations. It is suitable to systems that can be described using discrete-time models.
  One contribution of our work is to combine discrete-time models with continuous time models.

\section{Conclusion and future work}
\label{sec:conclu}

  We presented the correctness proof of a time-triggered interface implementing 
  the bit-clock synchronization mechanism of the 
  FlexRay standard for automotive systems. 
  This proof involves to simultaneously prove that the receiver keeps track of the correct 
  bits and that its hardware allows for a proper synchronization. This difficulty comes from the fact
  that the hardware controls the state machine which in turn controls the hardware.
  The bit-clock synchronization algorithm works by resetting a counter when detecting a synchronization sequence.
  This specific value is a crucial parameter. Our proof reveals the exact values of this counter that guarantee 
  reliable transmissions. This proves and disproves values proposed in the literature.
  Our proof is based on a general and precise model of asynchronous communications which includes 
  clock drift, clock jitter and metastability. The proof is performed using a hybrid 
  methodology that combines interactive reasoning in Isabelle/HOL and automatic model-checking 
  using NuSMV within Isabelle via the tool IHaVeIt. 
  
  Our model of asynchronous communications is very general.
  The model is about 2 000 lines of Isabelle/HOL code\footnote{See http://www.cs.ru.nl/$\sim$julien/Julien\_at\_Nijmegen/corr11.html}.
  It can be easily re-used. A user deals only with Proposition~\ref{thm:reg-affw}, Proposition~\ref{thm:input-values}, 
  and our drift assumption (Definition~\ref{eq:bcd}).
  The design that we presented and analyzed is part of a more complex system
  which includes a fault-tolerant scheduler. 
  Our model, its integration with IHaVeIt, and the supporting methodology 
  have been re-used to verify -- at the gate-level -- this fault-tolerant 
  scheduler~\cite{abk:memocode08}.

  The proof presented here was developed in about one man-year and is about 8 000 lines.
  Most of the time was spent developing the model and checking whether the model 
  was too weak or that there was an error in the design. We indeed discovered
  errors in early designs. 
  Most of the proof about our FlexRay-like interface is dedicated to the deduction 
  of valid digital inputs from the analog 
  transmission. This technique is independent of the design under verification. 
  If one would prove a similar design, one would be able to re-use most of 
  these lemmas. The main task would be to adapt the digital lemmas to this 
  new design. These lemmas would be proven automatically.
  We estimate that the time needed to develop such a new proof
  would be about a couple of weeks.

  An interesting future research direction would be to structure the proof in a way that will make this 
  separation between design-dependent lemmas and more general ones explicit. 
  To this end, we need to identify a set of constraints on the digital design 
  that would be sufficient to prove our final theorem. This would reduce 
  the analysis of similar designs to proving different instances of these \emph{digital} propositions.
  The theorem proving efforts are performed while formalizing computer architectures.
  The verification of particular designs reduce to discharging a set of constraints which are more likely
  to fall into their scope and limit the state-space explosion problem.

\section*{Acknowledgments}

Part of this work was carried out while the author 
was affiliated with the University of Saarland, Saarbr\"ucken, Germany. 
This work was funded by the German Federal Ministry of Education and Research (bmb+f)
in the framework of the Verisoft project under grant 01 IS C38.
This work initiated from the lecture ``Computer Architecture 2 -- Automotive Systems'' 
given by Paul at Saarland University and notes taken by students\footnote{www-wjp.cs.uni-sb.de/lehre/lehre.php}.

\bibliographystyle{plain}

\begin{thebibliography}{10}

\bibitem{PS05}
FlexRay Communication System -- Protocol Layer Specification v2.1, Rev A,
  FlexRay Consortium, December 2005.

\bibitem{abk:memocode08}
E.~Alkassar, P.~B{\"o}hm, and S.~Knapp.
\newblock Correctness of a fault-tolerant real-time scheduler and its hardware
  implementation.
\newblock In {\em Sixth ACM-IEEE International Conference on Formal Methods and
  Models for Codesign (MEMOCODE'08)}, pages 175--186. IEEE Computer Society,
  2008.

\bibitem{BK08}
C.~Baier and J.-P. Katoen.
\newblock {\em Principles of Model Checking (Representation and Mind Series)}.
\newblock The MIT Press, 2008.

\bibitem{AutoICCD05}
S.~Beyer, P.~B\"ohm, M.~Gerke, M.~Hillebrand, T.~In der Rieden, S.~Knapp,
  D.~Leinenbach, and W.~J. Paul.
\newblock Towards the formal verification of lower system layers in automotive
  systems.
\newblock In {\em ICCD '05: Proceedings of the 2005 International Conference on
  Computer Design}, 2005.

\bibitem{bosscher94verification}
D.~Bosscher, I.~Polak, and F.~W. Vaandrager.
\newblock Verification of an audio control protocol.
\newblock In {\em ProCoS: Proceedings of the Third International Symposium
  Organized Jointly with the Working Group Provably Correct Systems on Formal
  Techniques in Real-Time and Fault-Tolerant Systems}, pages 170--192, London,
  UK, 1994. Springer-Verlag.

\bibitem{brown_pike_06}
G.~M. Brown and L.~Pike.
\newblock Easy {P}arameterized {V}erification of {B}iphase {M}ark and {8N1}
  {P}rotocols.
\newblock In {\em The Proceedings of the 12th International Conference on Tools
  and the Construction of Algorithms ({TACAS'06})}, volume 3920 of {\em LNCS},
  pages 58--72, 2006.

\bibitem{NuSMV02}
A.~Cimatti, E.~M. Clarke, E.~Giunchiglia, F.~Giunchiglia, M.~Pistore,
  M.~Roveri, R.~Sebastiani, and A.~Tacchella.
\newblock {NuSMV} 2: An opensource tool for symbolic model checking.
\newblock In {\em Proceedings of the 14th International Conference on Computer
  Aided Verification (CAV'02)}, volume 2404 of {\em LNCS}, pages 359--364,
  Copenhagen, Denmark, July 27--31 2002. Springer.

\bibitem{CGP99}
E.~Clarke, O.~Grumberg, and D.~Peled.
\newblock {\em Model {C}hecking}.
\newblock {MIT} {P}ress, 1999.

\bibitem{SAL2}
L.~de~Moura, S.~Owre, H.~Rue{\ss}, J.~Rushby, N.~Shankar, M.~Sorea, and
  A.~Tiwari.
\newblock {SAL 2}.
\newblock In R.~Alur and D.~Peled, editors, {\em Computer-Aided Verification,
  {CAV} 2004}, volume 3114 of {\em LNCS}, pages 496--500, Boston, MA, 2004.
  Springer-Verlag.

\bibitem{hanna94reasoning}
K.~Hanna.
\newblock Reasoning about real circuits.
\newblock In {\em Proceedings of the 7th International Workshop on Higher Order
  Logic Theorem Proving and Its Applications}, pages 235--253, London, UK,
  1994. Springer-Verlag.

\bibitem{hanna98}
K.~Hanna.
\newblock Automatic verification of mixed-level logic circuits.
\newblock In {\em FMCAD '98: Proceedings of the Second International Conference
  on Formal Methods in Computer-Aided Design}, pages 133--166, London, UK,
  1998. Springer-Verlag.

\bibitem{KM00}
M.~Kaufmann, P.~Manolios, and J~Strother Moore.
\newblock {\em {ACL2} {C}omputer {A}ided {R}easoning: {A}n {A}pproach}.
\newblock Kluwer Academic Press, 2000.

\bibitem{maenner88}
R.~Manner.
\newblock Metastable states in asynchronous digital systems: Avoidable or
  unavoidable.
\newblock {\em Microelectronic Reliability}, 28(2):295--307, 1988.

\bibitem{BIPHIM93}
J~Strother Moore.
\newblock A {F}ormal {M}odel of {A}synchronous {C}ommunications and {I}ts {U}se
  in {M}echanically {V}erifying a {B}iphase {M}ark {P}rotocol.
\newblock {\em Formal Aspects of Computing}, 6(1):60--91, 1993.

\bibitem{IsabelleTutorial}
T.~Nipkow, L.C. Paulson, and M.~Wenzel.
\newblock {\em Isabelle/HOL: {A} {P}roof {A}ssistant for {H}igher-{O}rder
  {L}ogic}, volume 2283 of {\em LNCS}.
\newblock Springer, 2002.

\bibitem{Verilog}
V.~Sagdeo.
\newblock {\em The Complete VERILOG Book}.
\newblock Kluwer Academic Publishers, Norwell, MA, USA, 1998.

\bibitem{Al-SammaneZT07}
G.~Al Sammane, M.~H. Zaki, and S.~Tahar.
\newblock A symbolic methodology for the verification of analog and mixed
  signal designs.
\newblock In {\em DATE}, pages 249--254, 2007.

\bibitem{FMCAD06}
J.~Schmaltz.
\newblock A {F}ormal {M}odel of {L}ower {S}ystem {L}ayer.
\newblock In {\em Formal Methods in Computer-Aided Design (FMCAD'06)}, San
  Jose, CA, USA, November 12-16 2006. IEEE/ACM.

\bibitem{FMCAD07}
J.~Schmaltz.
\newblock A {F}ormal {M}odel of {C}lock {D}omain {C}rossing and {A}utomated
  {V}erification of {T}ime-{T}riggered {H}ardware.
\newblock In J.~Baumgartner and M.~Sheeran, editors, {\em Formal Methods in
  Computer-Aided Design (FMCAD'07)}, Austin, TX, USA, 11-14 November 2007.
  IEEE/ACM.

\bibitem{Tv09}
S.~Tverdyshev.
\newblock {\em Formal Verification of Gate-Level Computer Systems}.
\newblock PhD thesis, Saarland University, Computer Science Department, 2009.

\bibitem{Tverdyshev:TIME08}
S.~Tverdyshev and E.~Alkassar.
\newblock Efficient bit-level model reductions for automated hardware
  verification.
\newblock In S.~Demri and C.~S. Jensen, editors, {\em 15th International
  Symposium on Temporal Representation and Reasoning: TIME2008}, pages pp.
  164--172. IEEE Computer Society Press, 2008.

\bibitem{vaandrager06}
F.~W. Vaandrager and A.~de~Groot.
\newblock Analysis of a biphase mark protocol with uppaal and pvs.
\newblock {\em Formal Asp. Comput.}, 18(4):433--458, 2006.

\bibitem{AMSVerif08}
M.~H. Zaki, S.~Tahar, and G.~Bois.
\newblock Formal verification of analog and mixed signal designs: a survey.
\newblock {\em Microelectronics Journal}, 39:1395--1404, 2008.

\end{thebibliography}

\end{document}